%% file: main.tex
\documentclass{article}
\usepackage[utf8]{inputenc}

\usepackage{natbib} 
\input{packages}

\usepackage{authblk}

\usepackage[hmargin=0.75in,vmargin=1in]{geometry} 

\usepackage{shortcuts}
\input{paperSpecificNotation}

\title{Understanding the Risks and Rewards of Combining Unbiased and Possibly Biased Estimators, with Applications to Causal Inference}
\date{May 24th, 2023\footnote{A previous version of this working paper titled \enquote{Bias-robust Integration of Observational and Experimental Estimators} was presented at the American Causal Inference Conference 2022.}}

\author[1,2]{Michael Oberst}
\author[3]{Alexander D'Amour}
\author[3]{Minmin Chen}
\author[3]{Yuyan Wang}
\author[1,2]{David Sontag}
\author[3]{Steve Yadlowsky}
\affil[1]{MIT CSAIL}
\affil[2]{MIT Institute for Medical Engineering \& Science}
\affil[3]{Google DeepMind}

\begin{document}
\maketitle

\begin{abstract}
Several problems in statistics involve the combination of high-variance unbiased estimators with low-variance estimators that are only unbiased under strong assumptions. A notable example is the estimation of causal effects while combining small experimental datasets with larger observational datasets.  There exist a series of recent proposals on how to perform such a combination, even when the bias of the low-variance estimator is unknown.

To build intuition for the differing trade-offs of competing approaches, we argue for examining the finite-sample estimation error of each approach as a function of the unknown bias. This includes understanding the bias threshold --- the largest bias for which a given approach improves over using the unbiased estimator alone. Though this lens, we review several recent proposals, and observe in simulation that different approaches exhibits qualitatively different behavior.

We also introduce a simple alternative approach, which compares favorably in simulation to recent alternatives, having a higher bias threshold and generally making a more conservative trade-off between best-case performance (when the bias is zero) and worst-case performance (when the bias is adversarially chosen). More broadly, we prove that for any amount of (unknown) bias, the MSE of this estimator can be bounded in a transparent way that depends on the variance / covariance of the underlying estimators that are being combined. 
\end{abstract}

\section{Introduction}%
\label{sec:introduction}

We consider estimation of a real-valued quantity $\tz \in \R$ with small mean-squared error (MSE), in settings where we have access to both an unbiased estimator $\tuh$, typically with high-variance, and a possibly biased estimator $\tbh$ with low variance but unknown bias. This problem arises in a variety of settings in causal inference, as illustrated here:\footnote{A detailed treatment of~\cref{ex:combining_obs_exp_estimates,ex:combining_surrogates_and_primary_outcomes} is given in~\cref{sec:combining_ate_estimates_derived_from_observational_and_experimental_samples,sec:combining_surrogates_and_primary_outcomes_in_experimental_data}.}

\begin{example}[Combining observational and experimental estimates]\label{ex:combining_obs_exp_estimates}
  Given a randomized control trial (RCT) with a known probability of treatment, we can construct an unbiased estimator $\tuh$ under minimal assumptions. However, this estimator may have large variance / mean-squared error, if the sample size is small.  Given a larger observational dataset with the same treatment, we can often construct a lower-variance estimator $\tbh$ of the treatment effect.  However, this estimator will only be consistent under strong causal assumptions (e.g., no unmeasured confounding).
\end{example}

\begin{example}[Augmenting RCTs with observational controls]\label{ex:augment_rct_obs_controls}
  For a novel treatment, there may not be any treated units in observational data outside the RCT\@.  We can still use observational data to augment the RCT with additional control units, by e.g., using matching methods to choose observational controls that are similar to treated units in the RCT\@.  This will typically yield a lower-variance estimator $\tbh$ of the treatment effect, but a violation of relevant causal assumptions will lead to an unknown degree of bias.
\end{example}

\begin{example}[Making use of short-term surrogate outcomes]\label{ex:combining_surrogates_and_primary_outcomes}
  There are other ways to use observational data, even for a novel treatment not observed outside the RCT\@.  For instance, short-term surrogate outcomes $S$ may be available in the experimental data, which are believed to mediate the effect of treatment on the primary outcome of interest $Y$.  Under strong causal assumptions, one can use an estimator of $\E[Y \mid S]$ (e.g., trained on a large observational dataset) as a lower-variance version of $Y$ in the RCT, to construct an estimator $\tbh$ that is consistent if the assumptions hold, and otherwise potentially biased.
\end{example}

In all of the examples above, $\tbh$ is \textbf{potentially biased}, due to potential violations of our causal assumptions, but may have \textbf{lower variance} than the unbiased estimator. Our focus in this work is estimation\footnote{Given our focus on estimation, we do not consider questions of inference based on asymptotic properties (e.g., the construction of confidence intervals).}, with the goal of constructing an estimator $\thh$ of $\tz$ that has small mean-squared error $\E[(\thh - \tz)^2]$.  In this context, we hope to make a bias-variance trade-off by combining $\tbh$ and $\tuh$ in some way that performs better than using $\tuh$ alone. 

To this end, there are several recent proposals for adaptively combining unbiased and potentially biased estimators $\tuh$ and $\tbh$, often motivated by the combination of observational and experimental data (\cref{sec:motivation_and_setup}).  
We refer to these proposals as \textit{combined estimators}, as they represent strategies for combining estimators $\tuh, \tbh$ to form a new estimator.
These proposals implement intuitive heuristics, including adaptive linear combination strategies that approximate an optimal linear combination of $\tuh$ and $\tbh$ \citep{Cheng2021-sn}; hypothesis testing strategies that combine $\tuh$ and $\tbh$ only if a test fails to reject the hypothesis that $\tuh$ and $\tbh$ have the same mean \citep{Yang2020-na}; and soft-thresholding strategies that lie somewhere in between \citep{Chen2021-eo}.
Each of these strategies has been shown to have some favorable properties in certain asymptotic regimes.
However, there is some ambiguity about how these properties and regimes relate to one another, and what the implications are for practice in finite samples.

In this paper, we introduce a perspective that provides a more complete, unified view of these adaptive combination strategies.
A central motivation for our framing is that, while combined estimators can be advantageous in some settings, no combined estimator can dominate the unbiased estimator $\tuh$ alone under all levels of unknown bias (an issue that we explore in \cref{sec:impossibilities}).
Thus, when studying adaptive combined estimators, the key considerations are \emph{under what bias levels} the combined estimator's performance compares favorably to the unbiased one, as well as the risk / reward of using the combined estimator in an unfavorable / favorable regime.

\begin{figure}[t]
  \centering
  \includegraphics[width=0.35\linewidth]{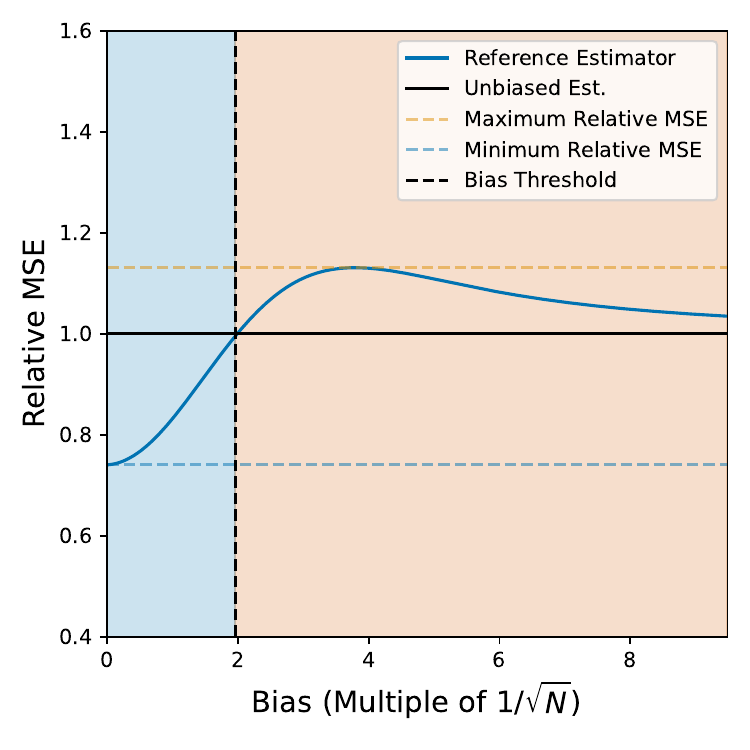}
  \caption{
    A \enquote{performance curve} for the reference estimator introduced in~\cref{sec:our_estimator}. 
    The y-axis shows the relative MSE\@: the ratio of the mean-squared error (MSE) of the reference estimator to the MSE of the unbiased estimator. The x-axis shows the bias $\mu$ of the biased estimator. 
    The curve highlights the best-case performance (minimum relative MSE), worst-case performance (maximum relative MSE), and bias threshold, with a large-bias rMSE limit of 1 as $\mu / \srn \rightarrow \infty$. 
    In this example, the biased $(\tbh)$ and unbiased $(\tuh$) estimators are independent with the same variance and sample size. Further details given in~\cref{sec:empirical}. 
  }%
  \label{fig:intro_example_curve}
\end{figure}

To study these properties, we introduce the \enquote{Performance Curve} of a combined estimator as a central tool for analysis.
Briefly, a performance curve charts how an estimator performs (in this paper, in terms of MSE) across different levels of unknown bias in an otherwise fixed setting.
\cref{fig:intro_example_curve} illustrates such a curve for a particular combined estimator.
Performance curves highlight several important practical properties of combined estimators: the best-case and worst-case performance, the large-bias limit, and the ``bias threshold'', or the level of bias above which the combined estimator under-performs the unbiased estimator $\tuh$ alone.
Previous analytical frameworks have addressed some parts of this curve under certain asymptotic regimes; however, we argue that performance curves provide a fuller picture, and often provide finite sample insights.

As a point of reference, we also introduce a particularly straightforward combined estimator that is a plug-in estimator of the MSE-minimizing weighted average of biased and unbiased estimators $\tbh$ and $\tuh$.
It is straightforward to analytically characterize several aspects of the plug-in estimator's performance curve, including the facts that (1) its worst-case performance is a bounded multiple of the unbiased estimator's MSE (we establish this with explicit constants); and (2) its performance converges to that of the unbiased estimator in the large bias regime.  The performance curve shown in \cref{fig:intro_example_curve} reflects this estimator, which we refer to as the \enquote{reference estimator}.

With performance curves and the reference estimator defined, we give a comparison of several recently proposed estimators (\cref{sec:empirical}), using simulation to construct their respective performance curves. This exercise illustrates that their behavior can be qualitatively different: For instance, not all estimators converge to the unbiased estimator as the bias grows.  Likewise, different estimators make different trade-offs between best-case performance and worst-case performance.
Interestingly, the reference plug-in estimator appears to be more robust than other estimators across the range of settings that we consider: its bias threshold is higher than the others under every simulation setting.

Finally, we illustrate how these ideas might be used in practice: In particular, they suggest that when the bias is likely to be substantial, a combination strategy may be ill-advised to begin with. To build intuition for how much bias is allowable, we construct a simulation to mimic the SPRINT Trial~\citep{SPRINT_Research_Group2015-gb}, alongside a much larger (but confounded) observational study.  Here, we simulate a curve analogous to~\cref{fig:intro_example_curve}, but directly parameterized by sensitivity analysis parameters, and find that the reference estimator can tolerate a moderate amount of influence of the confounder on treatment assignment.

\section{Statistical Framework}%

\subsection{Motivation and Setup: Biased and Unbiased Estimators}%
\label{sec:motivation_and_setup}
Abstractly, the statistical problem we study in this work is quite general, where we seek to estimate a parameter $\tz \in \R$ with minimum mean-squared error, 
\begin{equation}\label{eq:def_mse}
  \MSE(\thetah) \coloneqq \E\left[{(\thetah - \tz)}^2\right],
\end{equation}
by combining an estimator $\tuh$ with zero bias and an estimator $\tbh$ with an unknown bias. These estimators may or may not be independent of one another. The expectation in~\cref{eq:def_mse}, as well as other expectations discussed in this article, is taken with respect to sample draws used to construct both $\tuh$ and $\tbh$, and implicitly depends on the size of those samples. We denote the bias by $\mu$, where $E[\tuh] = \tz, \E[\tbh] = \tz + \mu$.  Note that the bias $\mu$ may implicitly depend on the sample size, but we do \textbf{not} assume that it converges to zero due to systematic bias.

We consider methods for combining the estimators $\tuh$ and $\tbh$ that take these estimates as input, as well as some relevant auxiliary statistics, such as the standard errors and covariance of the estimators. We use the following notation for the variance and covariance of the individual estimators: $\sigmabs \coloneqq \var(\tbh), \sigmaus \coloneqq \var(\tuh)$, and $\sigmacov \coloneqq \cov(\tuh, \tbh)$, which are generally estimated (directly or indirectly) as part of approaches which seek to combine the two estimators.\footnote{The specifics of estimating $\sigmauhs, \sigmabhs, \sigmacovh$ will vary based on the underlying estimators and application.  General approaches include the non-parametric bootstrap, or using consistent estimators of the asymptotic variance / covariance (scaled by $n^{-1}$) as an approximation of the finite-sample variance / covariance.  Estimators for the latter generally exist for regular and asymptotically linear (RAL) estimators.}

We focus on methods that are \emph{equivariant} in the estimators, meaning that if $(\tuh, \tbh; \dots) \mapsto \hat\theta(\tuh, \tbh; \dots)$, then $\hat\theta(\tuh + c, \tbh + c; \dots) = \hat\theta(\tuh, \tbh; \dots)$. This condition ensures that the method is not introducing bias of its own that makes the performance of the method depend on the true value of the parameter $\theta_0$, rather than only the estimation error of $\tuh$ and $\tbh$. Most of our analysis would directly extend to non-equivariant estimators, but would require specifying the choice of $\theta_0$ under which the analysis was performed. All of the methods from the literature that we study satisfy this equivariance property.

We evaluate estimators based on their mean squared error (MSE), noting that a central goal of estimator combination is to make a potentially favorable bias-variance trade-off that could reduce MSE\@.
For an estimator\footnote{The $\dots$ denotes any auxiliary information used by the estimator.} $(\tuh, \tbh; \dots) \mapsto \hat\theta$ and any fixed distribution $P$ for $\tuh$ and $\tbh$, the MSE for $\hat\theta$ is defined over draws of $\tuh$ and $\tbh$ from this distribution.
It is particularly useful to compare the combined estimator's MSE to that of the unbiased estimator $\tuh$ alone, which is the natural ``conservative'' alternative in most cases.
We define the relative MSE of a combined estimator $\hat \theta$ as
\begin{equation}
  \label{eq:relative_mse}
    \rMSE(\hat\theta) = \frac{\MSE(\hat\theta)}{\MSE(\tuh)}.
\end{equation}
An estimator $\thetah$ is said to out-perform / under-perform $\tuh$ when the relative MSE is less than 1 / greater than 1. 

\subsection{Limitations of Combined Estimators}%
\label{sec:impossibilities}
Ideally, there would exist an adaptive combination strategy for which $\rMSE(\thetah) < 1$, regardless of the unknown bias $\mu = \E[\tbh] - \tz$ of $\tbh$.
Unfortunately, this is not generally possible.  
Consider the following classic example, well known in the literature on shrinkage estimation.

\begin{example}[No Free Lunch]\label{ex:no_free_lunch}
  Let $X_1, \ldots, X_n \in \R$ be iid Gaussian samples\footnote{The one-dimensional restriction to $\tz \in \R$ in~\cref{ex:no_free_lunch} is a meaningful one: For estimating the mean of a multi-variate Gaussian with dimension $\geq 3$, where we consider the MSE over the entire vector $\|\thetah - \tz\|_2^2$, the sample average is dominated by shrinkage estimators, a fact exploited in recent work in causal inference \citep{Rosenman2020-cl}. However, even in these settings, the shrinkage estimator does not dominate the sample mean with respect to the MSE of any particular component of $\theta_0$.}, and consider estimation of $\tz \coloneqq \E[X]$. Let $\tuh$ be the sample average $n^{-1} \sum_{i=1}^{n} X_i$, and let $\tbh$ be some constant $c \in \R$ where we make no assumptions on the relationship between $c$ and $\tz$.  Here, $\tbh$ provides no information on $\tz$, and without further assumptions, we should not expect to improve upon $\tz$ by using $\tbh$ in some way.  More formally, $\tuh$ is admissible: In the one-dimensional setting, no estimator exists which \textbf{always} out-performs $\tuh$ in terms of MSE \citep{Stein1956-zl}.
\end{example}

When designing adaptive combined estimators $\thetah(\tuh, \tbh)$, this fact appears as the difficulty of estimating the unknown bias $\mu$ of $\tbh$.
Notably, all of the adaptive strategies that we review here estimate $\mu$ using $\tbh - \tuh$ as an unbiased estimator.
However, we can observe for independent $\tbh, \tuh$, estimating the bias of $\tbh$ is \textbf{at least as hard} as estimating the original parameter $\tz$, by the simple observation that $\var(\tbh - \tuh) = \var(\tbh) + \var(\tuh) \geq \var(\tuh)$.

The fact that a combined estimation strategy cannot, in general, dominate the unbiased estimator $\tuh$ is not necessarily a reason for pessimism.
Instead, knowing that combining estimators cannot be beneficial under all circumstances motivates understanding \emph{under which circumstances} a combined estimator could provide benefit.

\section{Performance Curves}

\subsection{Definition}
In this section, we introduce performance curves, which give a fine-grained view of the performance of a combined estimator $\thetah(\tuh, \tbh)$ across a range of scenarios for the bias of $\tbh$.
The performance curve plots $\rMSE(\thetah)$ across a family of distributions that are indexed by this unknown bias.

\begin{definition}
The \emph{performance curve} for an equivariant estimator combination method $\hat\theta(\tuh, \tbh; \dots)$ and a distribution $P$ over $(\tuh, \tilde{\theta}_b)$ with $\E_P[\tuh] = \E_P[\tilde{\theta}_b] = \theta_0$ is a plot of $\rMSE(\hat\theta)$ with respect to $\mu$, where the $\rMSE$ is with respect to the distribution over $(\tuh, \tbh)$ induced by drawing $(\tuh, \tilde{\theta}_b) \sim P$ and setting $\tbh = \tilde{\theta}_b + \mu$.
\end{definition}
Performance curves can be generalized to non-equivariant estimators by indexing the curve by the true parameter $\theta_0$, however such generalization is not necessary for the methods that we study in this work. 
\Cref{fig:perf-curve-anatomy} shows the performance curve for the reference estimator, defined in~\cref{sec:our_estimator}.
We recommend plotting performance curves with a horizontal line at 1, representing MSE equal to that of the unbiased estimator $\tuh$ alone, for reference.

The performance curve can be defined for any distribution $P$, and therefore, for any finite sample sizes for the data used to construct the unbiased and biased estimators. However, we pay special attention to jointly normal estimators $\tuh, \tbh$, for the following reason: If the respective sample sizes are substantial enough, it can be useful to approximate the distributions by an asymptotic distribution. For instance, if the estimators are jointly $\sqrt{n}$-consistent and asymptotically normal, then $P$ can be reasonably approximated as a Gaussian distribution with $\var( (\sqrt{n}(\tuh - \theta_0), \sqrt{n}(\tbh - \mu - \theta_0))^\top ) = \Sigma$. If we write the bias as $\E[ \tbh ] = \mu/\sqrt{n}$ (abusing notation by reusing $\mu$ as the scaled bias here), and plot the performance curve with the $\rMSE$ as a function of $\mu$, the performance curve will converge to a fixed object as $n \to \infty$ under reasonable regularity conditions on the combination method $\hat\theta(\tuh, \tbh; \dots)$. This asymptotic performance curve will correspond to the curve for choosing $P$ so that $(\tuh, \tbh)^\top \sim \normal{}( (\theta_0, \mu + \theta_0)^\top, \Sigma)$. 

A nice property of performance curves is that they can be instantiated concretely via simulation. For instance, in~\cref{sec:empirical} we use a simple simulation design where $\tuh$ and $\tbh$ are normally distributed, to approximate a large-sample scenario where $\tuh$ and $\tbh$ are asymptotically normal (similar to the motivation behind the asymptotic analysis in \citet{Yang2020-na}).  However, as we illustrate in~\cref{sec:semirealistic_synthetic_experiment}, other simulation designs are possible: A core thesis of this work is that a sequence of simulations (indexed by the unknown bias of $\tbh$) are a useful tool for gaining insights into the shape of the performance curves and the properties discussed in Section~\ref{sec:perf-curve-anatomy}.

\subsection{Anatomy of a Performance Curve}%
\label{sec:perf-curve-anatomy}

The performance curve features several key properties that matter when using an estimator combination method in practice.
Here, we discuss several of these properties, and what we might expect from ``reasonable'' combination methods.
These properties are summarized visually in~\cref{fig:perf-curve-anatomy}.

\begin{figure}[t]
\centering
  \begin{subfigure}[t]{0.45\textwidth}
  \centering
    \includegraphics[width=0.8\linewidth]{figs/intro-curve-scaled.pdf}
  \caption{}%
  \label{fig:intro-curve-scaled}
  \end{subfigure}
  \begin{subfigure}[t]{0.45\textwidth}
  \centering
    \includegraphics[width=0.8\linewidth]{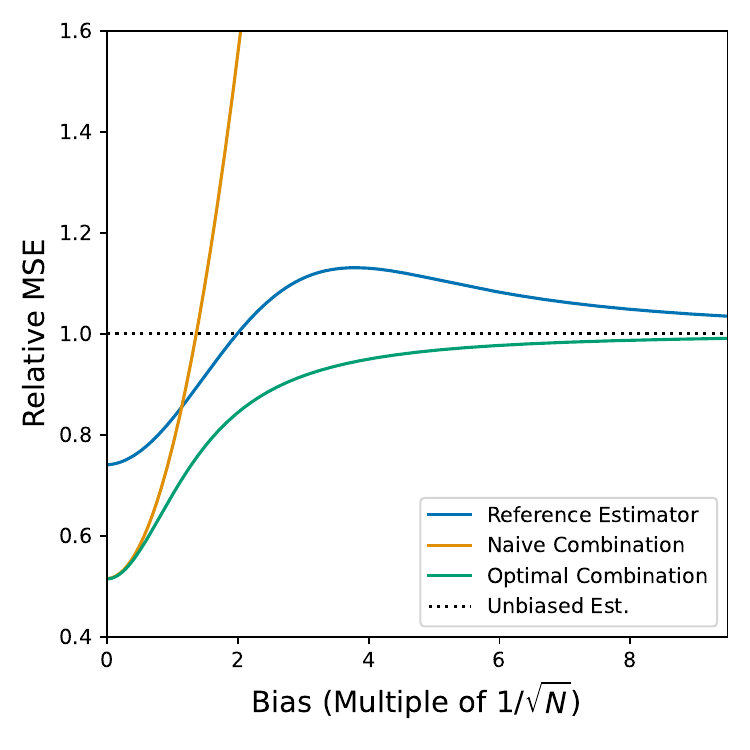}
  \caption{}%
  \label{fig:intro-curve-scaled-multiple}
  \end{subfigure}
  \caption{
    Anatomy of a performance curve: The y-axis shows the relative MSE (see~\cref{eq:relative_mse}). The x-axis shows the bias $\mu$ of the biased estimator. (\subref{fig:intro-curve-scaled}, also shown in~\cref{fig:intro_example_curve}) Performance curve for the reference estimator introduced in~\cref{sec:our_estimator}, highlighting the best-case performance (minimum relative MSE), worst-case performance (maximum relative MSE), and bias threshold, with a large-bias rMSE limit of 1 as $\mu / \srn \rightarrow \infty$. (\subref{fig:intro-curve-scaled-multiple}) Comparison against the performance curves for a naive variance-weighted combination (ignoring potential bias) and the optimal linear combination.
}%
\label{fig:perf-curve-anatomy}
\end{figure}

\paragraph{Best-case performance}
One useful reference for the best-case performance of an estimator is the minimax $\rMSE$, where $\inf_{\hat\theta} \sup_{\mu} \rMSE(\hat\theta)$ is bounded below by the $\rMSE$ of an optimistic estimator that correctly assumes $\mu = 0$,
\begin{equation*}
    \inf_{\hat\theta} \sup_{\mu} \rMSE(\hat\theta) \ge \inf_{\hat\theta} \sup_{\mu = 0} \rMSE(\hat\theta).
\end{equation*}
A simple corollary of standard results for the minimax error in mean estimation among location-equivariant estimators, and for combining two unbiased estimators is that the right hand side is attained by choosing $\hat\theta^\ast(\tuh, \tbh) = \lambdas \tbh + (1-\lambdas) \tuh$, with
\begin{equation}\label{eq:optimal_lambda_unbiased}
  \lambdas = \frac{\sigmaus - \sigmacov}{\sigmaus + \sigmabs - 2 \sigmacov}.
\end{equation}
Therefore, we expect the performance curve to always be above the $\rMSE$ of this estimator.\footnote{without making further (exploitable) assumptions about $\theta_0$, that can be exploited by $\hat\theta(\dots)$.}

\paragraph{Worst-case performance}
Unlike the best-case performance, which is bounded from below, there is no requirement that the worst-case is bounded from above. For instance, observe that $\rMSE(\tbh) \to \infty$ as $\mu \to \infty$. Therefore, a combination that ignores the potential for bias in $\tbh$ (e.g., by combining $\tbh$ and $\tuh$ using $\lambdas$ defined in~\cref{eq:optimal_lambda_unbiased}) can have unbounded error as a function of the bias (see~\cref{fig:intro-curve-scaled-multiple}: Naive combination). This is not to say that every method of combination will have unbounded worst-case performance. In fact, all of the methods considered in Section~\ref{sec:combination_methods} have bounded worst-case performance. However, without further information about $\mu$, the admissibility of $\tuh$ mentioned in Example~\ref{ex:no_free_lunch} implies that we should expect all performance curves to have some point with an $\rMSE > 1$, or be the constant $\rMSE = 1$ for all values of $\mu$. Therefore, we expect the worst-case $\rMSE$ to be above $1$.

\paragraph{Large-bias \rMSEtext{} Limit}
Many of the methods discussed in~\cref{sec:combination_methods} are adaptive to the bias, and mostly ignore $\tbh$ when the bias is large enough to easily tell from $\tbh$ and $\tuh$ alone; such methods have $\rMSE \to 1$ as $\mu \to \infty$. More generally, all of the approaches considered have a finite limit. Therefore, another natural property of a performance curve is the limiting $\rMSE$ as $\mu \to \infty$, which we will refer to as the large-bias \rMSEtext{} limit.

Additionally, while it is not guaranteed, all of the methods that we consider have an $\rMSE < 1$ for $\mu$ small enough. Then, as $\mu$ increases, it reaches a threshold where the $\rMSE$ crosses $1$, and exceeds the error of using the unbiased estimator on its own. We call this value of $\mu$ the \emph{bias threshold}. We can think of the bias threshold as a measure of the robustness of the combination method, as it corresponds to the maximum tolerance for bias for which the method improves upon the trivial baseline of using $\tuh$ alone.

\subsection{Comparison to Previous Analysis Approaches}
One of our main contributions is analysis of the entire performance curve for combination methods under parameterized distributions of the biased- and unbiased- input estimates. Such curves are easy to instantiate for specific methods through computer simulation, and provide rich insights for methodologists to consider when designing or choosing an approach for a given practical setting. Most existing analytical work on combining biased and unbiased estimators has focused on mathematical analysis of specific features of the performance curve. Mathematical approaches are advantageous in their generality, coming at the cost of the high-resolution afforded by mapping out the entire performance curve. Nonetheless, these mathematical approaches are interesting complementary modes of analysis to the one that we pursue, which provide useful perspectives in a variety of theoretical frameworks.

Common theoretical frameworks consider asymptotic arguments about combined estimators that typically consider bias regimes for $\tbh$, in which the bias' asymptotic order is compared to the asymptotic order of the variances $\tuh$, taken to be $O(n^{-1/2})$.
Thus, there are three bias regimes, in which estimators behave qualitatively differently: the low-bias regime in which the bias of $\tbh$ scales as $o(n^{-1/2})$; the high-bias regime in which the bias $\tbh$ is of higher order $\Theta(n^{-1/2})$; and the medium-bias or local asymptotic regime in which the bias is $\Omega(n^{-1/2})$. Estimators exhibit qualitatively different performance in these three regimes.

For example, \citet{Cheng2021-sn} focus specifically on the small-bias behavior of estimators, and demonstrate that their estimator is optimized the match the performance of the optimal combination of unbiased estimators in this regime. They also show that the estimator is consistent in the high-bias regime, showing that the large-bias \rMSEtext{} is bounded.

Meanwhile \citet{Chen2021-eo} consider two regimes;
see Theorem 3.2 of that work for a discussion of the phase transition that occurs between these two regimes. In one regime, the naively-pooled estimator is minimax optimal, and in the other, the RCT estimator is minimax optimal. The regime is determined based on whether or not an upper bound holds on the bias $\bar{\Delta} \lesssim \sigma_c / \sqrt{n_c}$, where $\sigma_c, n_c$ are the standard deviation and sample size of the RCT data.  They give an estimator that matches the performance of an oracle (which selects between these two) up to polylog factors (see Theorem 3.3).  However, the focus on rates ignores important constant factors that can be significant in practice. For example, estimators shown to be minimax optimal in their sense have significantly higher large-bias \rMSEtext{} limits in our simulations in Section~\ref{sec:combination_methods} (see Figure~\ref{fig:curves_anchor}).

So far, these results focus primarily on the behavior at small- or large- values of the bias. 
\citep{Yang2020-na}, on the other hand, focus on the medium bias regime via hypothesis testing with local asymptotic alternatives. In the context of our setting, this corresponds to a bias of $\mu / \srn$ for a fixed value of $\mu$.  They also derive the bias and MSE of their estimator under the null hypothesis that $\mu = 0$, and the fixed alternative where the bias does not scale with sample size, corresponding to the small-bias limit.

\citet{Dang2022-vd} refer to the local asymptotic regime as constituting intermediate bias, where they denote the bias as $\Psi_{s}^{\#}{(P_{0,n})}$ (see Eq. 3 of that work), and $\srn \Psi_{s}^{\#}{(P_{0,n})} \cip C$ for some constant $C$ (see Table 1).  They consider a setting where a \enquote{selector} is used to choose between different experiments (including potentially biased real-world data), where $s$ is used to denote the experiment (see the introduction to Section 3).  They derive the limiting distribution of the selector itself (see Table 1), as the minimizer of a quantity related to the MSE\@.

These regimes are also reflected in the performance curve.
However, the performance curve goes further, and highlights that a finer-grained analysis of the middle regime is particularly useful for making practical decisions about using a combined estimator in practice.
Specifically, on a performance curve, the low-bias regime appears as $\mu \rightarrow 0$, while the high-bias regime appears as $\mu \rightarrow \infty$.
The entirety of the curve in between corresponds to the middle-bias regime.
Importantly, the bias threshold at which the combined estimator underperforms the unbiased estimator alone, as well as the worst-case bias, both occur in this middle regime.
Thus, when making decisions about how to use a combined estimator in practice, the middle regime is the most relevant, and the specific value of the bias (not merely its order) plays a critical role. 
The goal of the performance curve is to bring this fine-grained structure of the problem front and center.

\section{Methods for Combining Estimators}

\subsection{A simple estimator for reference}%
\label{sec:our_estimator}

We first introduce a simple estimator, which does not require hyperparameter selection (in contrast to prior approaches, discussed in~\cref{sec:prior_approaches_for_combining_estimators}), and which has some favorable properties. In particular, we consider linear combinations of estimators of the form 
\begin{equation}\label{eq:estimator}
  \thl \coloneqq \lambda \tbh + (1 - \lambda) \tuh = \tuh + \lambda (\tbh - \tuh),
\end{equation}
for a real-valued weight $\lambda \in \R$, whose MSE depends on the trade-off between the relative variances of both estimators, as well as their covariance.  
Estimators of this general form have a long history in forecasting and model averaging, dating back to \citet{Bates1969-jc}, and variants of this linear combination strategy have been proposed in the context of estimating conditional average treatment effects using kernel regression \citep{Cheng2021-sn} and stratum-specific effects using shrinkage estimation \citep{Rosenman2020-cl}, as we discuss in~\cref{sec:prior_approaches_for_combining_estimators}.
The theoretically optimal weight $\lambdas$ for minimizing the MSE is given by\footnote{We give a short proof of this claim, which is a generally known fact, in~\cref{sec:proofs}.}
\begin{equation}\label{eq:optimal_lambda}
  \lambdas = \frac{\sigmaus - \sigmacov}{\mu^2 + \sigmaus + \sigmabs - 2 \sigmacov}.
\end{equation}
The optimal weight depends on several unknown quantities (including the bias) that must be estimated from data. The simple approach is to estimate $\lambdas$ using plug-in estimates of each quantity, using $(\tbh - \tuh)^2$ as an estimate of $\mu^2$, alongside plug-in estimates of the variance and covariance of $\tuh, \tbh$. We use $\thlh$ to denote the estimator with $\lambdah$ estimated in this fashion, and refer to this estimator as the \textbf{reference estimator} throughout. 
\begin{align}\label{eq:lambdah}
  \thlh &= \lambdah \tbh + (1 - \lambdah) \tuh &\text{where}&& \lambdah &= \frac{\sigmauhs - \sigmacovh}{{(\tuh - \tbh)}^2 + \sigmauhs + \sigmabhs - 2 \sigmacovh}.
\end{align}

\subsection{Other approaches for combining estimators}%
\label{sec:prior_approaches_for_combining_estimators}%
\label{sec:combination_methods}

Here, we review a few recent proposals for combining biased and unbiased estimators. Typically, these are motivated by scenarios where the estimators are independent (i.e., combining observational effect estimates with those of randomized trials), but they are straightforward to extend to the general case we consider here, where the estimators may be correlated. We defer more detail to~\cref{sec:additional_baseline_details}. For each set of estimators, we show an illustrative set of performance curves for a single data-generating process $P$, where $\tbh, \tuh$ are independent with $\tbh \sim \cN(\tz + \mu, 1/n), \tuh \sim \cN(\tz, 1/n)$ with $\tz = 1$ and $n = 1000$.

\begin{figure}[t]
\begin{subfigure}[t]{0.24\textwidth}
  \includegraphics[width=\linewidth, height=\linewidth, keepaspectratio]{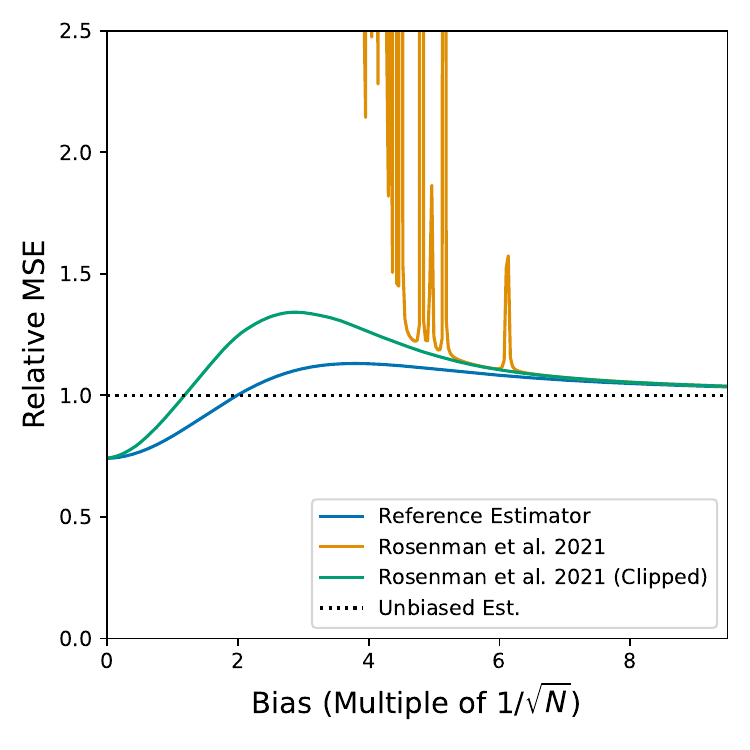}
  \caption{}%
  \label{fig:curves_rosenman}
\end{subfigure}
\begin{subfigure}[t]{0.24\textwidth}
  \includegraphics[width=\linewidth, height=\linewidth, keepaspectratio]{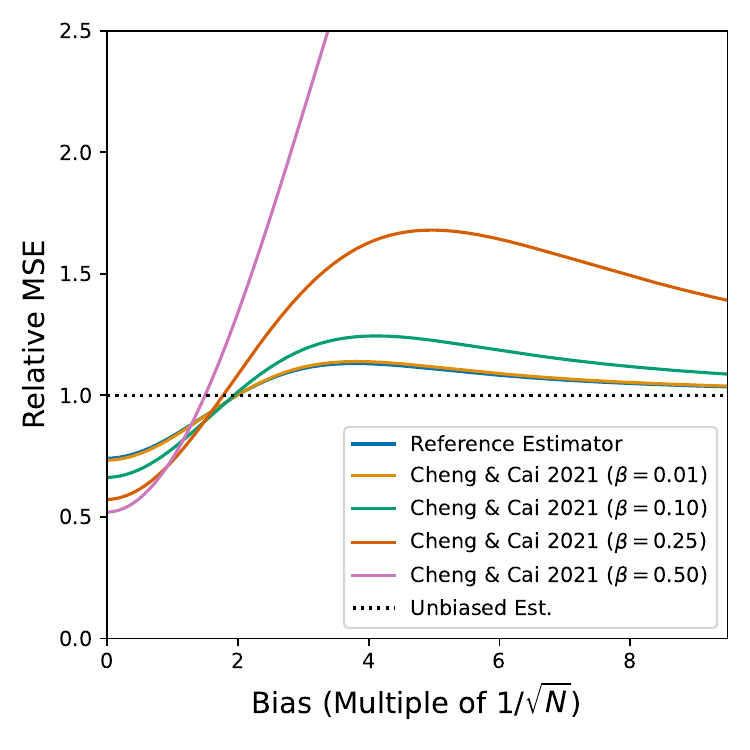}
  \caption{}%
  \label{fig:curves_cheng}
\end{subfigure}
\begin{subfigure}[t]{0.24\textwidth}
  \includegraphics[width=\linewidth, height=\linewidth, keepaspectratio]{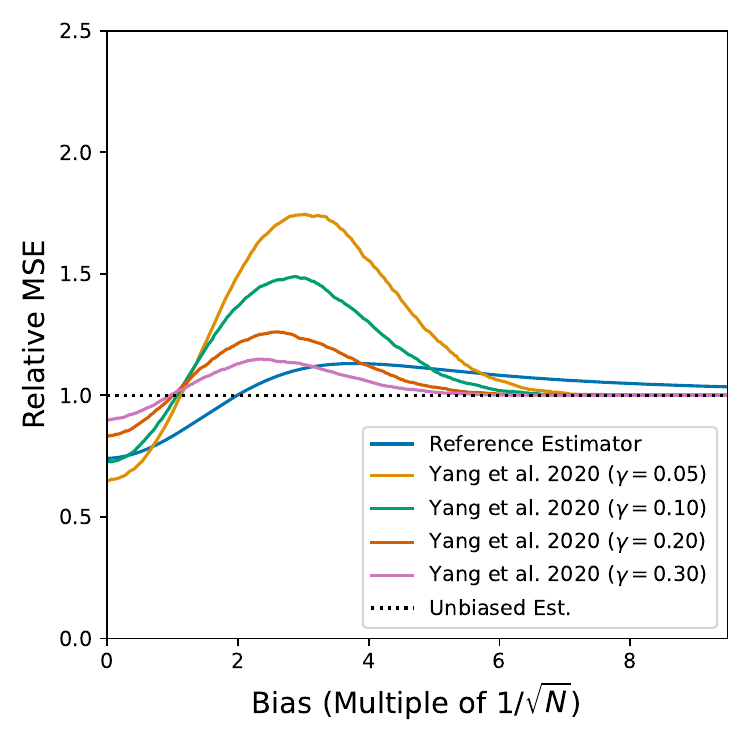}
  \caption{}%
  \label{fig:curves_ht}
\end{subfigure}
\begin{subfigure}[t]{0.24\textwidth}
  \includegraphics[width=\linewidth, height=\linewidth, keepaspectratio]{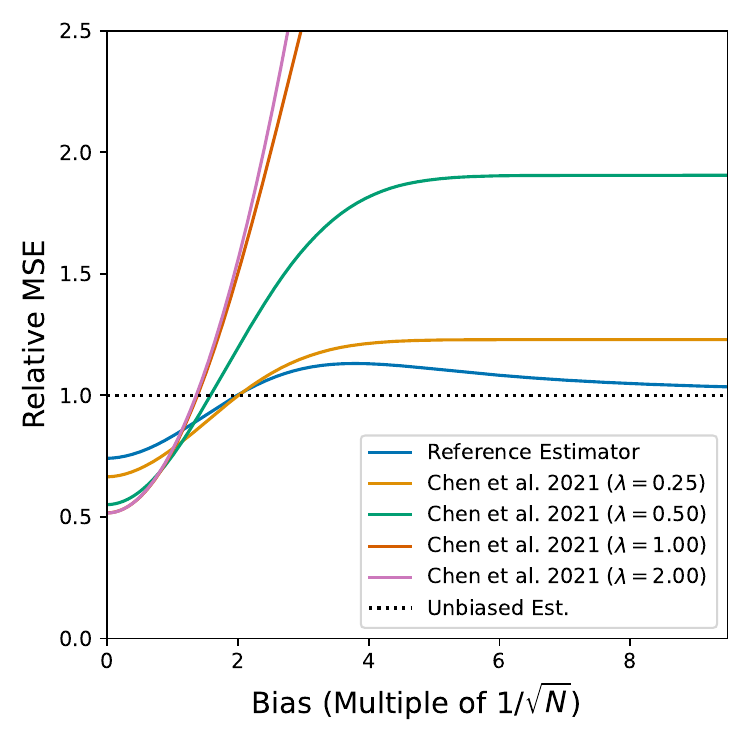}
  \caption{}%
  \label{fig:curves_anchor}
\end{subfigure}%
\label{fig:curves}
\caption{Performance curves reveal qualitative differences in behavior for different estimators, as well as the impact of different hyperparameter choices.  Here we give the performance curves for (\subref{fig:curves_rosenman}) \citet{Rosenman2020-cl}, with and without clipping (\subref{fig:curves_cheng}) \citet{Cheng2021-sn} across choices of hyperparameter $\beta$ (\subref{fig:curves_ht}) \citet{Yang2020-na} across choices of significance threshold $\gamma$ (\subref{fig:curves_anchor}) \citet{Chen2021-eo} across choices of hyperparameter $\lambda$.  We discuss these patterns in more detail in~\cref{sec:combination_methods}.}
\end{figure}
\textbf{Shrinkage}: \citet{Rosenman2020-cl} propose shrinkage estimators for combining (multivariate) observational and experimental estimators via a similar linear combination. While they consider independent estimators, their approach is straightforward to adapt to our setting.  In particular, they seek to estimate the optimal weighting ($\lambdas$ in~\cref{eq:optimal_lambda}), observe that $(\tuh - \tbh)^2$ is an unbiased estimate of the denominator, 
and propose to estimate $\lambdas$ by
\begin{equation}
  \lambdah = (\sigmauhs - \sigmacovh)/(\tuh - \tbh)^2, 
\end{equation}
which is optionally clipped to lie in $[0, 1]$. 
The main theoretical guarantees of \citet{Rosenman2020-cl} are not relevant in our setting, as they focus on the case where the target estimand is multivariate, allowing for the use of classical results \citep{Stein1981-az, Strawderman2003-gq} which give dominance in terms of MSE across the entire vector.

In~\cref{fig:curves_rosenman} we show the performance curve of the original estimator and the clipped estimator, observing that both have a similar large-bias limit of 1, but that clipping plays an instrumental role in reducing the rMSE for smaller values of bias.  In fact, in this simulated setting, clipping uniformly improves the performance of the estimator, which otherwise never achieves an rMSE less than 1.

\textbf{Adaptive Linear Combination}: \citet{Cheng2021-sn} propose an adaptive linear combination of estimators that is similar to the reference estimator. 
Their focus is on CATE estimation with kernel regression, in the context of combining experimental and trial estimators. 
However, for the setting we consider (estimating a real-valued parameter), their estimator reduces to a similar linear combination, where 
\begin{equation}
  \lambdah  = \frac{\sigmauhs - \sigmacovh}{n^{-\beta} {(\tuh - \tbh)}^2 + \sigmauhs + \sigmabhs - 2 \sigmacovh}
\end{equation}
where $\beta > 0$ is a hyperparameter. Note that this differs from~\cref{eq:lambdah} due to the factor of $n^{-\beta}$ in the denominator.
The theoretical results presented in \citet{Cheng2021-sn} focus on consistency and adaptivity, showing that when the bias is consistent regardless of the bias, and that if the bias is zero, $\lambdah$ converges to the optimal inverse-variance weights.  This latter property is the motivation for including the $n^{-\beta}$ term as part of the estimated bias,\footnote{This estimator can also be motivated as a form of ridge regression with a weighted $\ell_2$ penalty.  \citet{Cheng2021-sn} also propose an estimator that is analogous to a weighted $\ell_1$ penalty, which we do not discuss here.} but introduces difficulties in choosing an appropriate value of $\beta$.

In~\cref{fig:curves_cheng}, we observe the general shape of the performance curve, as well as the impact of varying the hyperparameter $\beta$.  For smaller values of $\beta$, the approach is more conservative, with lower worst-case rMSE, and higher best-case rMSE\@.  Larger values of $\beta$ are less conservative, achieving better performance in the zero-bias regime, but with the potential for substantially worse performance when the bias is large (see the curve for $\beta = 0.5$). For every hyperparameter setting, the rMSE converges to 1 in the large-bias limit.

\textbf{Hypothesis Testing}: \citet{Yang2020-na} give a procedure that first tests for bias, pooling observational and experimental data if this test fails to reject, and which otherwise uses only the experimental data. The details of this test depend on the underlying estimators,\footnote{We give a more detailed treatment in~\cref{sec:baselines}.} but when $\tuh$ and $\tbh$ correspond to sample averages, this reduces to a simple form: The test statistic is given by $T_n = (\tuh - \tbh)^2/(n \sigmah)$, where $\sigmah$ is an estimate of the standard deviation of $\tuh - \tbh$, such that $T_n$ follows a chi-square distribution under the null hypothesis that the bias is zero. 
Their analysis focuses on the asymptotic properties of this test-based procedure, including the asymptotic regime where the bias scales as $n^{-1/2}$, where each estimator is asymptotically normal, and where the scaled errors follow a limiting mixture distribution.  Notably, they characterize the asymptotic bias and MSE in this regime as a function of both the unknown bias and a fixed choice of threshold for the hypothesis test.  The authors suggest using the estimated bias to select the threshold, though their asymptotic analysis does not handle this case.

In~\cref{fig:curves_ht}, we plot the performance curves of this approach for different values of $\gamma$, the significance threshold for rejecting the null hypothesis that the estimators share a common limit.  Higher values of $\gamma$ imply a higher likelihood of rejection, and hence more conservative performance, with higher worst-case rMSE and lower best-case rMSE\@.  For every hyperparameter setting, the rMSE converges to 1 in the large-bias limit, but tends to zero fairly quickly, as the probability of rejection goes to 1.

\textbf{Anchored Thresholding}: \citet{Chen2021-eo} attempt to estimate and correct for the bias in $\tbh$, estimating it via soft-thresholding as 
\begin{equation}
  \hat{\mu} = \begin{cases}
    \text{sign}({\tbh} - {\tuh})\left(\abs{{\tbh} - {\tuh}} - \lambda
    \sqrt{\hat{\var}({\tbh} - {\tuh})}\right), &\ \text{if } \abs{{\tbh} - {\tuh}}
    \geq \lambda \cdot \sqrt{\hat{\var}( {\tbh} - {\tuh} )} \\
    0, &\ \text{otherwise.}
  \end{cases}
\end{equation}
and then combine $\tuh$ and $\tbh - \hat{\mu}$ using inverse variance-weighting. 
\citet{Chen2021-eo} demonstrate that this estimator achieves the performance (up to poly-log factors) of an oracle that selectively chooses an estimator based on whether or not (in our notation) the bias $\mu$ is larger or smaller than $\sigmau$, and demonstrate that this performance is minimax optimal under a certain data-generating process.
In contrast to their analysis, which hides constant factors, our investigation in~\cref{sec:empirical} focuses on understanding the constant factors involved in a finite-sample setting, and understanding qualitative performance across different values of the bias. However, we do not make any claims about performance in the setting where the observational data is of a higher order than the experimental data, which is a main focus of their work. 

In~\cref{fig:curves_anchor}, we plot the performance curves of this approach for different values of $\lambda$, observing qualitatively different behavior than the other estimators considered here.  In particular, the large-bias rMSE limit is not 1, but rather a value determined by $\lambda$.  This result follows from the fact that soft-thresholding is applied to the estimated bias itself, so that large estimated values of bias are always shrunk towards zero.  For smaller values of $\lambda$ (i.e., less aggressive shrinkage towards zero in the estimated bias), the approach is more conservative, with higher worst-case rMSE and lower best-case rMSE\@.  In contrast to the other estimators considered here, the large-bias rMSE limit is not 1, but rather the worst-case rMSE\@.

\section{Comparing Performance Curves}%
\label{sec:empirical}

Having investigated the performance curves of each estimator in the previous section with a single distribution $P$ to build intuition, we conduct a larger-scale comparison of the properties of each estimator across a broader range of settings.  We use the reference estimator $\thlh$ as a common point of comparison to the approaches discussed in~\cref{sec:prior_approaches_for_combining_estimators}.

\begin{table}[t]
  \centering
  \caption{In the experiments in~\cref{sec:empirical}, we evaluate each combined estimator across each of the following parameter settings. For each value of $\mu$, the estimators $\tuh, \tbh$ are constructed from drawing $n$ samples of $(\psiu, \psib)$ from a normal distribution with mean $(\tz, \tz + \mu)$, marginal variances $\var(\psiu)$, $\var(\psib)$ and covariance determined by the chosen correlation.  The value of $\tz$ is 1 throughout.}\label{tab:simulation_parameter}
  \begin{tabular}{lr}
    \toprule
    Parameter & Values \\
    \midrule
    $n$ & $\{500, 1000, 2000, 4000\}$ \\
    $\var(\psiu)$&$\{1, 2, 4, 8, 16\}$ \\
    $\var(\psib)$&$\{0, 1, 2, 4, 8, 16\}$\\
    $\text{corr}(\psiu, \psib)$&$ \{-0.5, -0.25, 0, 0.25, 0.5\}$\\
    $\mu$ & $[0, 1.5]$, increments of 0.002\\
    \bottomrule
  \end{tabular}
\end{table}

\textbf{Setup}: We let $\psiu, \psib$ be drawn from a multivariate normal distribution, where $\tuh = \frac{1}{n} \sum_{i} \psiu^{(i)}$ and $\tbh$ are the sample averages, and where we can directly estimate quantities like the variance $\var(\psiu)$. We then investigate the performance of each combined estimator for different variances / covariances of $(\psiu, \psib)$, as we vary the bias. In particular, we compute the squared error of $\thetah$ and the squared error of $\tuh$, and for each set of simulation parameters in~\cref{tab:simulation_parameter}, we repeat this process 10000 times to estimate the MSE\@. For each simulation setting, we compute the performance curves, as well as the above notable properties of these curves by sweeping over the value of $\mu$. This can be done efficiently, i.e., without resampling 10000 random variables again, by simply adding a variety of offsets to the previously sampled $\tbh$ to sweep over $\mu$.

We then compare the reference estimator to the approaches described in~\cref{sec:prior_approaches_for_combining_estimators}, with the following additional details. For the shrinkage estimator of \citet{Rosenman2020-cl}, we clip the weights to lie in $[0, 1]$, having observed in~\cref{sec:prior_approaches_for_combining_estimators} that this is necessary to get stable results for small values of the bias.
\citet{Yang2020-na} propose a data-adaptive approach to choosing the significance level in their test-based procedure, by estimating the bias directly and then simulating from the asymptotic mixture distribution of their estimator under that bias to select a cutoff that yields optimal performance.  We replicate this data-driven approach in our experiments, as described in~\cref{sec:additional_baseline_details}, but report results for different fixed thresholds in~\cref{sec:comparison_to_hypothesis_testing}.
For the remaining estimators, which require a choice of hyperparameter, we choose single value for simplicity, but report additional comparisons to other hyperparameters in~\cref{sec:comparison_of_different_hyperparameter_settings}.
For the adaptive linear combination proposed by \citet{Cheng2021-sn}, we use $\beta = 0.25$.
For the anchored thresholding approach of~\citet{Chen2021-eo}, we use $\lambda = 0.5 \sqrt{\log n}$, in keeping with their synthetic experimental setup.

\textbf{Summary of results}: For the simulation parameters we consider, the reference estimator has a higher bias threshold than alternative approaches, and no alternative estimator dominates in terms of better best-case and worst-case performance. Here, we focus on the relative performance of each approach, but we present additional results in~\cref{sec:understanding_factors_that_drive_improvement_and_} that shed light on the factors which drive the best and worst-case trade-off for the reference estimator.

\textbf{Building intuition with a single setting}:  In Figure~\ref{fig:single_example_curve_baseline}, we plot the performance curve of each approach on a common figure, for same parameter settings used in~\cref{sec:prior_approaches_for_combining_estimators} where $\var(\psiu) = \var(\psib) = 1, n = 1000, \corr(\psiu, \psib) = 0$. 

First, as observed previously, each estimator makes a trade-off between the worst-case and best-case relative MSE\@. Here, the reference estimator has among the lowest worst-case relative MSE of any approach, comparable to that of the hypothesis-testing estimator.  Second, we observe that the bias threshold for the reference estimator is higher than that of the alternative estimators, occurring at around $\mu = 0.06$ (around $2 / \srn$), while the value of $\mu$ that attains the worst-case relative MSE falls in the range $0.10$ to $0.15$ ($3/\srn$ to $5/\srn$) for all estimators except for the anchored thresholding estimator: In this range the squared bias of $\tbh$ is of the same order as the variance of the unbiased estimator.  Intuitively, in this regime the bias is sufficiently large that it introduces additional MSE, but small enough that it is difficult to detect.

\begin{figure}[t]
\centering
  \begin{subfigure}[t]{0.3\textwidth}
    \centering
    \includegraphics[width=\linewidth, height=\linewidth, keepaspectratio]{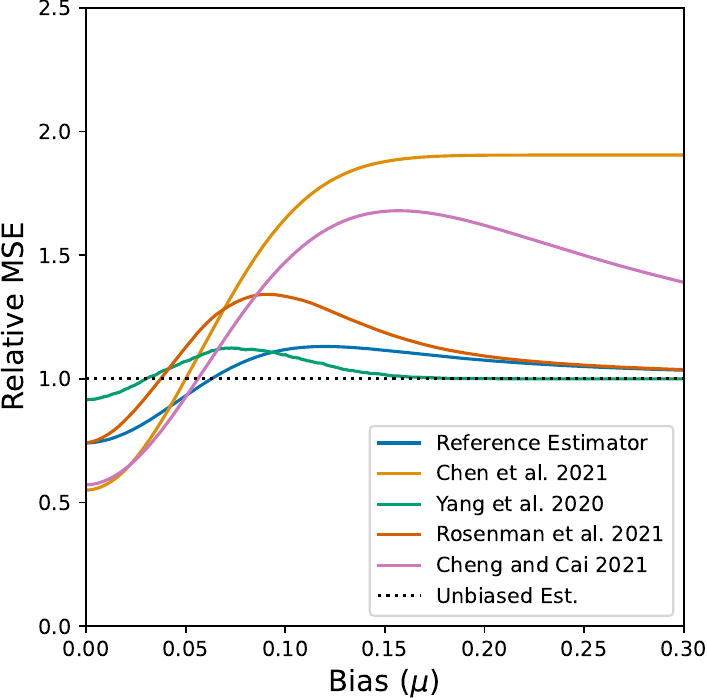}
    \caption{}%
    \label{fig:single_example_curve_baseline}
  \end{subfigure}
  \begin{subfigure}[t]{0.3\textwidth}
    \centering
    \includegraphics[width=\linewidth, height=\linewidth, keepaspectratio]{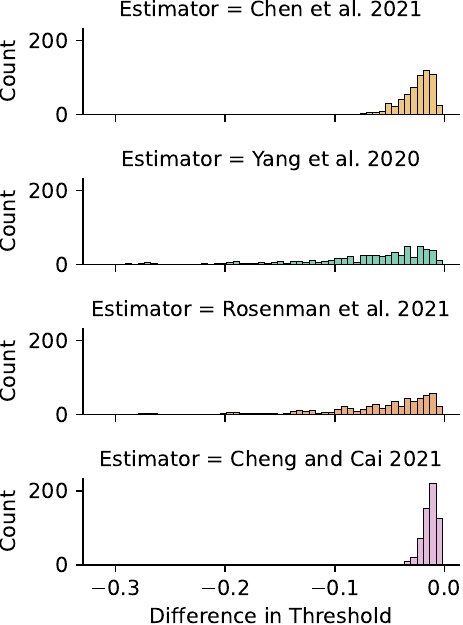}
    \caption{}%
    \label{fig:histogram_thresholds}
  \end{subfigure}
  \begin{subfigure}[t]{0.3\textwidth}
    \centering
    \includegraphics[width=\linewidth, height=\linewidth, keepaspectratio]{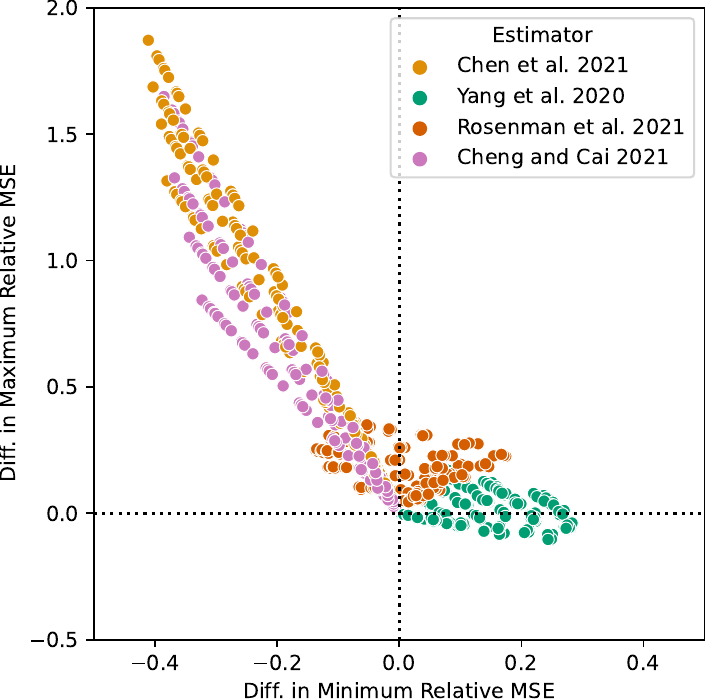}
    \caption{}%
    \label{fig:best_worst_curve_diff}
  \end{subfigure}
  \caption{
    (\subref{fig:single_example_curve_baseline}) For a fixed set of parameters from Table~\ref{tab:simulation_parameter}, we show $\mu$ on the x-axis, and $\MSE(\thetah) / \MSE(\tuh)$ on the y-axis (lower is better), for each estimator. 
    (\subref{fig:histogram_thresholds}) For each simulation setting, we compute the maximum value of $\mu$ (the \enquote{bias threshold}) for which the reference estimator has lower MSE than $\tuh$, and do similarly for each alternative estimator.  We then compute the difference of these thresholds, where a \textbf{negative} value means that the threshold of the reference estimator is higher (better), and plot a histogram of these differences across all simulations.
    (\subref{fig:best_worst_curve_diff}) For each simulation setting, we plot (on the y-axis) the maximum relative MSE of $\thetah$ for each alternative approach, minus the maximum relative MSE of $\thlh$, and (on the x-axis) we similarly plot the difference in the minimum relative MSE, where negative values (in both cases) signify an improvement over $\thlh$.}%
  \label{fig:synthetic_simulation}
\end{figure}

\textbf{Comparing curves across a variety of settings}: In~\cref{fig:histogram_thresholds,fig:best_worst_curve_diff}, we plot results over all parameter settings given in~\cref{tab:simulation_parameter}.  In~\cref{fig:histogram_thresholds}, we observe that the reference estimator has the highest bias-tolerance of the estimators (and hyperparameter settings%
\footnote{While we use the \enquote{default} hyperparameters for the estimator of \citet{Chen2021-eo} here, there are other values for which their bias threshold is marginally higher than that of the reference estimator, see~\cref{sec:comparison_of_different_hyperparameter_settings}.}%
) considered here.
We show differences of bias thresholds for each alternative estimator, where a \textbf{negative} value means that the threshold of the reference estimator is higher, and where each value corresponds to a different simulation. The threshold of the reference estimator is the highest in every scenario, indicated by fact that all reported values are negative.

In~\cref{fig:best_worst_curve_diff}, we observe that none of the alternative estimators dominates the reference estimator (achieving best and worst-case MSE that are \textbf{both} lower than that of the reference estimator) in any of the simulated scenarios. For each simulation setting, we plot (on the y-axis) the maximum relative MSE of $\thetah$ for each alternative approach, minus the maximum relative MSE of $\thlh$, and (on the x-axis) we similarly plot the difference in the minimum relative MSE, where negative values (in both cases) signify an improvement over $\thlh$.  Some estimators can be observed to make consistently different trade-offs:  In particular, the anchored-threshold estimator of \citet{Chen2021-eo} always has a higher maximum relative MSE, as well as a lower minimum relative MSE\@.  This may reflect the optimism of the approach: If the observed difference $\tuh - \tbh$ is sufficiently small, it assumes that the bias is equal to zero. Note that in some simulation scenarios (the upper right quadrant), the reference estimator has both a lower best-case and worst-case relative MSE\@. 

\section{Theoretical Bounds on Worst-Case Performance}%
\label{sec:theoretical}

So far, we have only investigated the performance of each estimator in simulation, which demonstrated a degree of robustness inherent to each approach, each of which appeared to have bounded estimation error for arbitrary values of the bias.  In this section, we demonstrate that this pattern holds more generally for the reference estimator. First, we observe that the reference estimator is consistent for the underlying causal effect (e.g., in the large-sample limit of experimental data), regardless of the bias, a property shared with the hypothesis-testing estimator of \citet{Yang2020-na} and the adaptive linear combination estimator of \citet{Cheng2021-sn}
\begin{restatable}[Consistency]{theorem}{Consistency}\label{thm:consistency}
  Let $\tbh \cip \tz + \mu$, $\tuh \cip \tz$, and let $n \sigmauhs \cip \nu_u$, $n \sigmabhs \cip \nu_b$, and $n \sigmacovh \cip \nu_{bu}$ for finite constants $\nu_u, \nu_b$, and $\nu_{bu}$.  Then, if $\mu \neq 0$, we have it that $\lambdah \cip 0$ and $\thlh \cip \tz$. 
\end{restatable}
All proofs can be found in~\Cref{sec:proofs}.  The intuition is straightforward: When the bias is non-zero, the term $(\tuh - \tbh)^2$ converges to a non-zero constant, while all other terms in $\lambdah$ converge to zero.\footnote{When $\tuh$ and $\tbh$ are asymptotically normal estimators, then the conditions of~\cref{thm:consistency} are easily satisfied, as discussed in~\cref{sec:proof_of_minor_claims}.} 

However, such asymptotic results do not say much about the impact of bias in finite samples. Motivated by this, we bound the worst-case behavior of $\thlh$ under arbitrary values of the bias, where we assume that $\tuh$ is unbiased. We additionally assume away some pathological cases where $\lambdah$ is not well-defined, or is trivially zero.
\begin{assumption}[Unbiased Experimental Estimator]\label{asmp:unbiased_u}
  The estimator $\tuh$ is unbiased, i.e., $\E[\tuh] = \tz$
\end{assumption}
\begin{assumption}[Non-Zero Variance]\label{asmp:nonzero-variance-diff}
  The unbiased estimator has non-zero variance $\sigmaus > 0$, and the difference $\tuh - \tbh$ has non-zero variance
  \begin{equation}
  \var(\tuh - \tbh) = \sigmaus + \sigmabs - 2 \sigmacov > 0
  \end{equation}
  and the estimators $\sigmauhs, \sigmabhs, \sigmacovh$ satisfy $\sigmauhs + \sigmabhs - 2 \sigmacovh > 0$.
\end{assumption}
Assumption~\ref{asmp:nonzero-variance-diff} rules out the case where $\tbh = \tuh + \mu$, and ensures that $\lambdah$ is always well-defined for all $\mu$, including $\mu = 0$.
Our main result is Theorem~\ref{thmthm:mse_bound_cov_unknown_var}, which bounds the relative MSE of the reference estimator by a constant factor that depends on the behavior of the estimators $\sigmauhs, \sigmabhs, \sigmacovh$.
\begin{restatable}[Bound on MSE]{theorem}{MSEBoundCovUnknownVar}\label{thmthm:mse_bound_cov_unknown_var}
  Under Assumptions~\ref{asmp:unbiased_u} and~\ref{asmp:nonzero-variance-diff}, with estimators $\sigmabhs, \sigmauhs, \sigmacovh$ that have bounded second moments, the MSE of the reference estimator $\thlh$ is bounded by 
  \begin{equation*}
    \E[{(\thetahlh - \tz)}^2] \leq {\left(\sigmau + \frac{1}{2} \sqrt{\E[S^2]}\right)}^2
  \end{equation*}
  where $\E[S^2] = \E[{(\sigmauhs - \sigmacovh)}^2 / (\sigmauhs + \sigmabhs - 2 \sigmacovh)]$.
\end{restatable}
To build intuition for the behavior of the worst-case bound, we can also state the following corollary, which gives the bound in terms of the underlying variance / covariance of the estimators $\tuh, \tbh$, if those quantities are known.
\begin{restatable}[Bound on MSE with known variance/covariance]{corollary}{MSEBoundCov}\label{thmcorr:mse_bound_cov}
  Under Assumptions~\ref{asmp:unbiased_u} and~\ref{asmp:nonzero-variance-diff}, and where $\sigmaus, \sigmabs, \sigmacov$ are known, define $c, \rho$ by $c \coloneqq \sigmab / \sigmau$ and $\rho = \sigmacov / \sqrt{\sigmaus \sigmabs}$, where $\rho = 0$ if $\sigmabs = 0$.  Then the MSE of the estimator $\thlh$ is bounded by 
  \begin{equation}
    \E[{(\thetahlh - \tz)}^2] \leq \sigmaus {\left(1 + \frac{1}{2} \frac{\abs{1 - \rho c}}{\sqrt{1 - 2 \rho c + c^2}}\right)}^2
  \end{equation}
\end{restatable}
\Cref{thmcorr:mse_bound_cov} gives us the intuition that the worst-case bound is largest when the biased estimator has favorable variance properties. For instance, if the estimators are independent $(\rho = 0)$, the worst-case relative MSE is a function of $c = \sigmab / \sigmau$, and the upper bound is \textbf{larger} when the variance $\sigmabs$ of the biased estimator is \textbf{smaller}. Intuitively, this reflects the fact that $\lambdah$ is larger (for any fixed bias) when the variance of the biased estimator is small.  Notably, the maximum value of this worst-case bound for $\rho = 0$ occurs when $c = 0$, and is equal to $2.25 \sigmaus$, a relatively small multiple of the MSE of the unbiased estimator, considering that it holds for any bias.

The intuition behind the proof of \cref{thmthm:mse_bound_cov_unknown_var,thmcorr:mse_bound_cov} is also instructive: We imagine an adversary who replaces the value of $\tbh$ with an adversarially chosen value. The optimal adversarial choice is to place $\tbh$ within one standard deviation (of the difference $\tbh - \tuh$) of $\tuh$, in a direction that pulls the estimate $\thlh$ away from the value of $\tz$ (see~\cref{lemma:MSEBoundCov_supremum}). Indeed, for large values of the bias, $\lambdah$ will tend to zero, and $\thlh$ will tend to the unbiased estimator $\tuh$, as formalized in the following.
\begin{restatable}{proposition}{ConvergeUnboundedBias}\label{prop:converge_unbounded_bias}
 Consider a sequence of biased estimators $\tbhk$ which can be written as $\tbh' + \mu_k$, where $\E[\tbh'] = \tz$, where $\cov(\tbhk, \tuh) = \sigmacov$, and $\var(\tbhk) = \sigmabs$.  Let $\mu_k \rightarrow \infty$ as $k \rightarrow \infty$. The MSE of the resulting sequence of estimators $\thetahlhk$ converges to the MSE of the unbiased estimator $\tuh$
\begin{equation}
  \lim_{k \rightarrow \infty} \E\left[{\left(\thetahlhk - \tz\right)}^2\right] = \sigmaus
\end{equation}
where $\thetahlhk = \lambdah \tbhk + (1 - \lambdah) \tuh$, and where $\lambdah = (\sigmaus - \sigmacov) / ({(\tuh - \tbh)}^2 + \sigmaus + \sigmabs - 2 \sigmacov)$.
\end{restatable}
\Cref{prop:converge_unbounded_bias} formally demonstrates that one of the patterns observed in our simulations (that the relative MSE converges to 1 as the bias grows without bound) holds more generally.

\section{Using Simulation to Assess the Maximum Allowable Bias}%
\label{sec:semirealistic_synthetic_experiment}

While~\cref{sec:theoretical} provides bounds on worst-case performance, we may wonder whether the bias thresholds are high enough in practice to warrant application of the method.  In practice, we recommend simple simulations to build intuition on this point, prior to using a combination strategy.  To illustrate, we construct a simulation where the parameters are designed to mimic the observed statistics of the SPRINT Trial \citep{SPRINT_Research_Group2015-gb}.  In this simulated scenario, we observe that the reference estimator out-performs the unbiased estimator when the confounding bias parameter $\gamma$ in the Rosenbaum sensitivity model \citep{Rosenbaum2010-qw} is less than 1.

The SPRINT Trial investigated the effectiveness of two different targets for systolic blood pressure ($<$120mm Hg, the \enquote{intensive} treatment, and $<$140mm Hg, the \enquote{standard} treatment) among non-diabetic patients with high cardiovascular risk.  We take $T = 0$ to denote the standard regime, and $T = 1$ to denote the intensive regime. Several of the outcomes considered in this trial are time-to-event outcomes: The primary composite outcome is comprised of myocardial infarction, other acute coronary syndromes, stroke, heart failure, or death from cardiovascular causes. For simplicity, we consider this outcome as a binary variable.  We take $Y = 1$ to denote the presence of the primary composite outcome. Note that in this simulation, the true value of the treatment effect is $\tz = \E[Y_1 - Y_0] = -0.0164$, a decrease of 1.64\% in the absolute risk of the primary outcome, chosen to match the statistics observed in the trial.  Here we use standard potential outcome notation, where $Y_t$ represents the potential outcome under treatment $t$.

\textbf{Creating an observational dataset with realistic confounding}: To construct confounded observational data, we first define a unobserved confounder. We use the reported trial statistics to calibrate the strength of the association between this confounder and the potential outcomes.  The trial reports the incidence of the primary outcome across both arms for several sub-groups (see Figure 4 of \citet{SPRINT_Research_Group2015-gb}). To emulate a plausible binary confounder, we consider previous chronic kidney disease (CKD), which has the smallest p-value for an interaction effect.  Taking $U = 1$ as presence of previous CKD, we then take the observed incidence of $Y$ in treatment and control, across these two subpopulations (from Figure 4 of \citet{SPRINT_Research_Group2015-gb}), as the values of $\E[Y_1 \mid U]$ and $\E[Y_0 \mid U]$ in our simulation.  We provide additional details in~\Cref{sec:sprint_simulation_details}. 

Given a pre-defined effect of $U$ on $Y_t$, we introduce confounding in the observational study via the following model for treatment selection $\P(T = 1 \mid U, D = O) = \text{logit}^{-1}(\gamma (U - 1/2))$, where we use $D = O$ to denote a data-point drawn from the observational study, and where the intercept is chosen to keep the log-odds symmetric around 0 for $U = 1, U = 0$. For $\gamma > 0$, patients with a history of CKD are more likely to receive intensive management.  There are no other covariates for simplicity.  This model of confounding can be viewed in the Rosenbaum sensitivity model \citep[See 4.2 of ][]{Rosenbaum2010-qw}, satisfying the bound
\begin{equation}
  \Gamma^{-1} \leq \frac{\P(T = 1 \mid U = 1, D = O) \P(T = 0 \mid U = 0, D = O)}{\P(T = 0 \mid U = 1, D = O) \P(T = 1 \mid U = 0, D = O)} \leq \Gamma
\end{equation}
with $\Gamma = \exp(\gamma)$.

\textbf{Simulation of estimator performance}: Based on the generative model above, we simulate $\nexp = 9361$ samples from the simulated trial (the size of the original trial), and a ten-fold larger amount from an observational study, $\nobs = 100000$.  We note that the generative model is identical (e.g., the distribution of $U$) except for the treatment assignment mechanism, and we examine the performance of the reference estimator as we vary the confounding bias $\gamma$.  In each dataset, the estimators $\tuh, \tbh$ are constructed by standard propensity score adjustment, with details of variance estimation given in~\cref{sec:sprint_simulation_details}. For each value of $\gamma$, we repeat this process 10000 times, where each iteration gives us one observation of the squared error for the reference estimator and the RCT estimator. We perform this procedure for 20 values of $\gamma$ evenly spaced between $0$ and $2$. For each value of $\gamma$, we obtain the corresponding value of $\Gamma$ as $\Gamma = \exp(\gamma)$.

\textbf{Results}: In Figure~\ref{fig:sprint_simulation_gamma} we compare the root mean-squared error (RMSE) of $\tuh$ and $\thlh$, for each value of $\gamma$.  In this particular scenario, we see that $\thlh$ improves on the performance of the unbiased estimator $\tuh$ in the regime where $\gamma < 1$, and otherwise tends to perform similarly. For reference, the impact of CKD on the composite outcome in the control group, also measured on the log-odds scale, is 0.55. In~\cref{sec:additional_sprint_results} we additionally vary the sample size $\nobs \in \{10000, 20000, 50000, 100000\}$, and observe that the maximum allowable value of $\gamma$ decreases slightly as the sample size increases. 

\begin{figure}[t]
  \centering
  \includegraphics[width=0.7\linewidth]{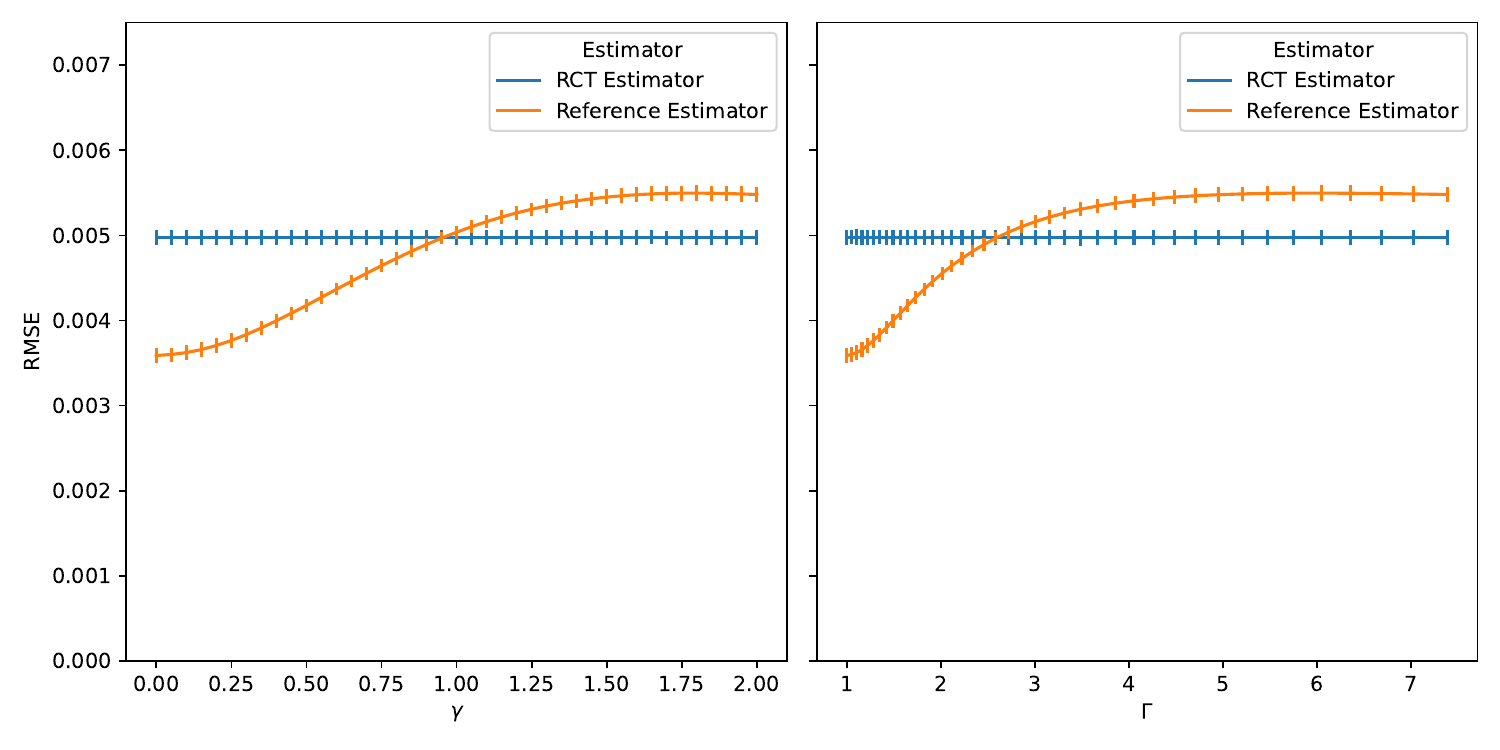}
  \caption{For each value of $\gamma$, and corresponding value of $\Gamma$, we show the root mean-squared error (RMSE) for each estimator, calculated over 10k simulations. 95\% confidence intervals are obtained via bootstrapping.}%
  \label{fig:sprint_simulation_gamma}
\end{figure}

\section{Discussion}%
\label{sec:discussion}

Many estimators are only consistent under strong assumptions, which we rarely believe to hold exactly in practice.  Given a potentially biased (but likely low-variance) estimator, and an unbiased (but likely high variance) estimator, we can seek a combination that performs better than the unbiased estimator alone.  We have discussed several recent proposals for doing so, in the context of causal inference, and discussed a simple baseline (the reference estimator) which requires no hyperparameter tuning, and which has reasonable worst-case guarantees on estimation performance.

We also introduced a different conceptual perspective for evaluating the use of such estimators, examining their finite-sample behavior as a function of the (unknown) bias, and we give a detailed simulation-based investigation of how several estimators perform along these lines.  In general, we advocate for the use of simple simulations in practice, to assess the relative trade-offs (and bias tolerance) of different estimators on data that is designed to resemble the problem at hand.

\paragraph{Acknowledgements} MO and DS were supported in part by Office of Naval Research Award No. N00014-21-1-2807. We thank Avi Feller and participants at the American Causal Inference Conference 2022 for providing feedback on an earlier version of this manuscript. 

\clearpage
\appendix
\section*{Appendix}%
\label{sec:appendix}

\section{Details of Motivating Examples}%
\label{sec:details_of_motivating_examples}

In~\cref{sec:an_unbiased_experimental_estimator} we describe an illustrative unbiased estimator, and in~\cref{sec:combining_surrogates_and_primary_outcomes_in_experimental_data,sec:combining_ate_estimates_derived_from_observational_and_experimental_samples} we describe two different estimators that make use of additional observational samples.

\textbf{Notation} We use $A$ to denote a binary action, $Y$ an outcome, $Y_a$ a potential outcome, and $X$ covariates.  We use $V$ to denote the full set of observed variables, such that we may have $V_i = (A_i, X_i, Y_i)$.  We use script characters to indicate the support of a random variable, e.g., $\cA = \{0, 1\}$. We use $\Dobs = \{V_i\}_{i=1}^{\nobs}$ to denote an observational sample of size $\nobs$, and $\Dexp = \{V_i\}_{i=1}^{\nexp}$ to denote an experimental sample of size $\nexp$, with a total sample size of $\nobs + \nexp = n$.  We use the random variable $D$ to denote the population from which a sample is drawn, with $D = O$ indicating membership in an observational sample, and $D = E$ indicating membership in the experimental sample. 

\subsection{An Unbiased Experimental Estimator}%
\label{sec:an_unbiased_experimental_estimator}

We assume throughout that there exists a consistent estimator $\tuh$ of the causal effect $\tz$, derived entirely from experimental data. For simplicity in the motivating examples that follows, we consider estimation of the causal effect in the experimental population
\begin{equation}\label{eq:causal_effect}
  \tz = \E[Y_1 - Y_0 \mid D = E],
\end{equation}
and assume that experimental sample is randomized with a fixed (known) probability of treatment.
\begin{assumption}[Identification of ATE in the Experimental Sample]\label{assmp:experimental_validity}
  The following hold:
  \begin{enumerate*}[label= (\roman*)]
    \item Consistency: $Y = Y_a$ when $A = a, D = E$,
    \item Ignorability: $Y_a \indep A \mid D = E$,
    \item Positivity: $\P(A = a \mid D = E) > 0$ for all $a \in \{0, 1\}$.
  \end{enumerate*}
\end{assumption}
Under Assumption~\ref{assmp:experimental_validity}, the causal effect can be estimated in an unbiased fashion using
\begin{equation}
  \tuh = \frac{1}{n} \sum^{n}_{i=1} \frac{\1{D_i = E}}{\hat{P}(D = E)} \left( \frac{Y_i A_i}{e} - \frac{Y_i (1 - A_i)}{1 - e}\right)
\end{equation}
where $e$ is the (known) probability of treatment assignment in the experimental sample, and $\hat{P}(D = E)$ is the empirical estimate $\nexp / (\nexp + \nobs)$.\footnote{Note that we write this as an average over the entire sample of size $\nexp + \nobs$, to be consistent with later notation, but the above is equivalent to taking the sample average of the pseudo-outcome over the experimental sample alone.} We consider this formulation for simplicity, because it yields an experimental estimator that is not only consistent, but is unbiased in finite samples.

\subsection{Combining Surrogates and Primary Outcomes in Experimental Data}%
\label{sec:combining_surrogates_and_primary_outcomes_in_experimental_data}

There are several settings in which we can construct an alternative estimator based on observational data.  We give one example here, and another in~\cref{sec:combining_ate_estimates_derived_from_observational_and_experimental_samples}.

Suppose that we have additional surrogate outcomes $S$ in the experimental sample, and an observational sample $\Dobs = {\{s_i, x_i, y_i\}}_{i=1}^{\nobs}$ containing surrogates, covariates, and outcomes, but no information on treatment.  We write $S_a$ to denote the potential surrogate outcome under treatment $A = a$. In this context, we can construct an alternative estimator of the causal effect, by using a \textit{surrogate index} estimator \citep{Athey2019-km} that leverages the observational sample to learn the causal relationship between the short-term surrogate outcomes $S$ and the long-term outcome $Y$. In contrast to the setting of \citet{Athey2019-km}, we assume that the outcome $Y$ is available in both experimental and observational samples, and consider using the surrogate only to obtain a higher-precision estimator.  For such an approach to yield an unbiased estimator, we require a few assumptions.
\begin{assumption}\label{assmp:surrogate_validity}
  The following conditions hold, in addition to those of Assumption~\ref{assmp:experimental_validity}:
  \begin{enumerate*}[label= (\roman*)]
    \item Unconfounded Treatment Assignment: $A \indep (Y_1,Y_0,S_1,S_0) \mid D = E$,
    \item Surrogacy: $A \indep Y \mid S, X, D = E$,
    \item Comparability: $D \indep Y \mid S, X$,
    \item Overlap: $\P(D = E \mid S = s, X = x) > 0$ for all $s \in \cS, x \in \cX$.
  \end{enumerate*}
\end{assumption}
Under Assumption~\ref{assmp:surrogate_validity}, which corresponds to Assumptions 1--4 of \citet{Athey2019-km}, the following quantities are equivalent in the experimental sample 
\begin{equation}
  \E[Y_a \mid D = E] = \E[Y \mid A, D = E] = \E[ h(S, X) \mid A, D = E]
\end{equation}
where $h(S, X) \coloneqq \E[Y \mid S, X, D = O]$ is referred to as the surrogate index \citep[See Theorem 1 of ][]{Athey2019-km}. Assumption~\ref{assmp:surrogate_validity} is reflected in Figure~\ref{subfig:surrogacy_graph}, which captures the fact that $Y$ is conditionally independent of both $A$ and $D$ given $S, X$, implying that there is no \textit{direct effect} of treatment on the outcome $Y$.

\begin{figure}[t]
\centering
  \begin{subfigure}[t]{0.5\textwidth}
    \centering
    \begin{tikzpicture}[
      obs/.style={circle, draw=gray!90, fill=gray!30, very thick, minimum size=5mm}, 
      uobs/.style={circle, draw=gray!90, fill=gray!10, dotted, minimum size=5mm}, 
      bend angle=30]
      \node[obs] (A) {$A$} ;
      \node[obs] (D) [below=of A] {$D$} ;
      \node[obs] (S) [right=of A]  {$S$};
      \node[obs] (Y) [right=of S]  {$Y$};
      \node[obs] (X) [below=of S] {$X$} ;
      \draw[-latex, thick] (X) -- (Y);
      \draw[-latex, thick] (X) -- (S);
      \draw[-latex, thick] (A) -- (S);
      \draw[-latex, red] (A) to[bend left] (Y);
      \draw[latex-latex, dotted, red] (S) to[bend right] (Y);
      \draw[-latex, thick] (S) -- (Y);
      \draw[-latex, thick] (D) -- (S);
      \draw[-latex, thick] (D) -- (X);
    \end{tikzpicture}
    \caption{}%
    \label{subfig:surrogacy_graph}
  \end{subfigure}%
  \begin{subfigure}[t]{0.5\textwidth}
    \centering
    \begin{tikzpicture}[
      obs/.style={circle, draw=gray!90, fill=gray!30, very thick, minimum size=5mm}, 
      uobs/.style={circle, draw=gray!90, fill=gray!10, dotted, minimum size=5mm}, 
      bend angle=30]
      \node[obs] (A) {$A$} ;
      \node[obs] (D) [below=of A] {$D$} ;
      \node[obs] (Y) [right=of A]  {$Y$};
      \node[obs] (X) [below=of Y] {$X$} ;
      \draw[-latex, thick] (X) -- (Y);
      \draw[-latex, thick] (X) -- (A);
      \draw[-latex, thick] (A) -- (Y);
      \draw[-latex, thick] (D) -- (A);
      \draw[-latex, red, thick] (D) -- (Y);
      \draw[latex-latex, red, dotted] (A) to[bend left] (Y);
      \draw[-latex, thick] (D) -- (X);
    \end{tikzpicture}
    \caption{}%
    \label{subfig:obs_exp_graph}
  \end{subfigure}
  \caption{(\subref{subfig:surrogacy_graph}) A causal graph consistent with Assumption~\ref{assmp:surrogate_validity}, where the red arrow denotes an illustrative edge that is \textbf{not} permitted: A direct effect of the treatment $A$ on the outcome $Y$, and the dotted bi-directional arrows illustrates another potential violation of assumptions due to unmeasured confounding. (\subref{subfig:obs_exp_graph}) A causal graph consistent with Assumption~\ref{assmp:internal_external_validity}, where the red arrows similarly indicate violations of the assumption.}%
\label{fig:}
\end{figure}
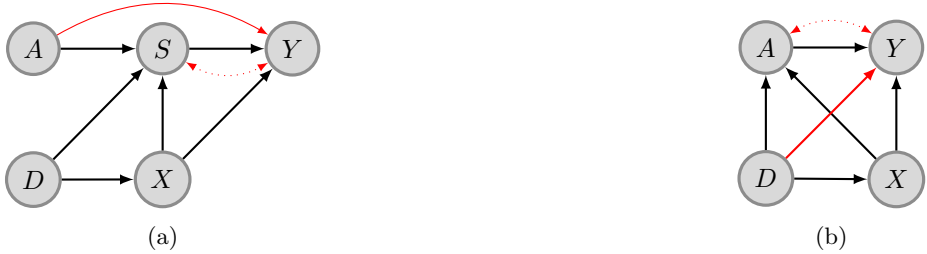

\begin{figure}[t]
\end{figure}
This provides us several possible methods for estimating the causal effect using a combination of the observational and experimental sample.  Here we give one simple estimator as an illustrative example
\begin{equation}
  \tbh = \frac{1}{n} \sum^{n}_{i=1} \frac{\1{D_i = E}}{\hat{P}(D = E)} \frac{\hat{h}(S_i, X_i) A_i}{e} - \frac{\hat{h}(S_i, X_i) (1 - A_i)}{1 - e}
\end{equation}
which takes the same form as $\tuh$, with $Y_i$ is replaced by $\hat{h}(S_i, X_i)$. Here if $\hat{h}(S_i, X_i)$ is a consistent estimator for the conditional expectation $h(S, X)$, and if Assumption~\ref{assmp:surrogate_validity} holds, then $\tbh$ is a consistent estimator of $\tz$.  We refer to $\tbh$ in this section as the surrogate index estimator. Asymptotic bias in $\tbh$ can arise due to violations of Assumption~\ref{assmp:surrogate_validity}, as illustrated in Figure~\ref{subfig:surrogacy_graph}, e.g., if there is a direct effect of the treatment $A$ on the outcome $Y$, or unmeasured confounding between the surrogates and outcome. 

\subsection{Combining ATE estimates from Observational and Experimental Samples}%
\label{sec:combining_ate_estimates_derived_from_observational_and_experimental_samples}

Alternatively, suppose we have access to a much larger observational study (with the same treatment, outcome, and covariates), whose support covers the RCT population, and which can be used to estimate the ATE in the RCT-population $\tz$ (Eq.~\ref{eq:causal_effect}).\footnote{For simplicity in these examples, we do not consider the causal effect in the observational population $\E[Y_1 - Y_0 \mid D = O]$, which requires assumptions about external validity of the RCT (e.g., that potential outcomes are conditionally independent of $D$ given $X$, that overlap holds between the studies, etc).  Since we focus on the setting where one estimator is known with high confidence to be consistent or unbiased, we focus on the setting where the causal effect of interest is defined with respect to the RCT population $\E[Y_1 - Y_0 \mid D = E]$, requiring no additional assumptions on the RCT population beyond Assumption~\ref{assmp:experimental_validity}.} In particular, if the causal effect can be identified in the observational study (internal validity) and transported to the RCT population (external validity), then we can hope to use the observational data to construct an alternative estimator of $\tz$, which might be expected to have smaller variance, especially if $\nobs$ is substantially larger than $\nexp$. These assumptions are formalized in Assumption~\ref{assmp:internal_external_validity}.
\begin{assumption}[Internal and External Validity of Observational Study]\label{assmp:internal_external_validity}
  The following conditions hold, in addition to those of Assumption~\ref{assmp:experimental_validity}:
  \begin{enumerate*}[label= (\roman*)]
    \item Consistency: $A = a, D = O \implies Y_a = Y$,
    \item Unconfounded Treatment Assignment: $Y_a \indep A \mid X, D = O$, $\forall a \in \cA$,
    \item Positivity of Treatment Assignment: $\P(X = x \mid D = O) \implies$ $\P(A = a \mid X = x, D = O) > 0$, $\forall a \in \cA, x \in \cX$,
    \item Unconfounded Selection: $Y_a \indep D \mid X$, $\forall a \in \cA$,
    \item Positivity of Selection: $\P(X = x) > 0 \implies$ $\P(D = d \mid X = x) > 0$, $\forall d \in \{O, E\}, x \in \cX$.
  \end{enumerate*}
\end{assumption}
The particulars of Assumption~\ref{assmp:internal_external_validity} can be found in the literature on transportability of causal effects (see \citet{Degtiar2021-bb} for a recent review).  Under Assumption~\ref{assmp:internal_external_validity}, the observational data can be used to estimate the causal effect using a variety of estimators.  For instance, one can estimate the causal effect as 
\begin{equation}
  \tbh = \frac{1}{n} \sum^{n}_{i=1} \frac{\1{D_i = E}}{\hat{P}(D = E)} \times (\hat{g}_1(X_i) - \hat{g}_0(X_i))
\end{equation}
where $\hat{g}_a(X)$ is an estimate of the conditional expectation under treatment in the observational population $g_{a}(X) \coloneqq \E[Y \mid A = a, D = O, X]$ \citep{Dahabreh2020-ua}. One can alternatively construct re-weighting estimators that use only the observational dataset, or doubly-robust estimators that combine the two (see Section 5 of \citet{Dahabreh2020-ua} for examples). Here we note that, with the exception of re-weighting estimators that use only the observational sample, the estimators $\tuh$ and $\tbh$ are not independent, due to the shared use of experimental samples. Note that $\tbh$ may not be consistent for $\tz$ if Assumption~\ref{assmp:internal_external_validity} fails, e.g., due to confounding in treatment assignment or selection.

\section{Proofs}%
\label{sec:proofs}

\subsection{Proof of Minor Claims}%
\label{sec:proof_of_minor_claims}

In~\Cref{sec:motivation_and_setup}, we claimed that $\lambda^*$ takes a particular form.  For completeness, we provide a proof here, but this is a known fact from the literature.

\begin{proof}
  Suppose that $\E[\tbh] = \tz + \mu$, and that $\E[\tuh] = \tz$. The corresponding MSE is given by 
\begin{align}
  &\text{MSE}(\lambda \tbh + (1 - \lambda) \tuh) \\
  &= \E[{(\lambda (\tbh - \tz) + (1 - \lambda) (\tuh - \tz))}^2]\\
  &= \E[\lambda^2 {(\tbh - \tz)}^2 + {(1 - \lambda)}^2 {(\tuh - \tz)}^2 + 2 \lambda(1 - \lambda) (\tbh - \tz) (\tuh - \tz)]\\
  &= \lambda^2 (\mu^2 + \var(\tbh)) + {(1 - \lambda)}^2 \var(\tuh) + 2 \lambda (1 - \lambda) \cov(\tbh, \tuh) \label{eq:quadratic_mse_in_lambda}
\end{align}
where we observe that $\E[{(\tbh - \tz)}^2] = \E[{(\tbh - \tb + \tb - \tz)}^2] = \mu^2 + \var(\tbh)$, and where we use the fact that 
\begin{align*}
  \E[(\tbh - \tz)(\tuh - \tz)] &= \E[(\tbh - (\tz + \mu))(\tuh - \tz) + \mu(\tuh - \tz)] \\
                           &= \E[(\tbh - (\tz + \mu))(\tuh - \tz)] + \E[\mu(\tuh - \tz)] \\
                           &= \E[(\tbh - (\tz + \mu))(\tuh - \tz)] \\
                           &= \cov(\tbh, \tuh)
\end{align*}
Note that Equation~\eqref{eq:quadratic_mse_in_lambda} is a quadratic in $\lambda$, which is minimized by setting the derivative equal to zero
\begin{align}
  0 &= 2 \lambda (\mu^2 + \var(\tbh)) - 2 (1 - \lambda) \var(\tuh) + (2 - 4\lambda)\cov(\tbh, \tuh) \\
    &= 2 \lambda (\mu^2 + \var(\tbh)) +2 \lambda \var(\tuh) - 4 \lambda \cov(\tbh, \tuh) - 2 \var(\tuh) + 2\cov(\tbh, \tuh) \\
  \implies \lambda^* &= \frac{\var(\tuh) - \cov(\tbh, \tuh)}{\mu^2 + \var(\tbh) + \var(\tuh) - 2\cov(\tbh, \tuh)} \label{eq:optimal_lambda_unobservable}
\end{align}
\end{proof}

In~\cref{sec:theoretical}, we claim that the conditions of~\cref{thm:consistency} are easily satisfied if $\tuh$ and $\tbh$ are regular and asymptotically normal estimators.  We formalize this in the following assumption.

\begin{assumption}[Asymptotic Regularity Conditions]\label{asmp:asymptotic_linearity}
  We assume that the estimators have an asymptotically linear representation
  \begin{equation}
    \srn\begin{pmatrix}
      \tuh - \tz \\ \tbh - (\tz + \mu)
    \end{pmatrix} = \frac{1}{\srn} \sum^{n}_{i=1} \begin{pmatrix}
      \infu(Z_i) \\ \infb(Z_i)
    \end{pmatrix} + o_p(1)
  \end{equation}
  where the random vector $(\infu(Z), \infb(Z))$ has zero-mean and bounded covariance, and where $\mu$ is a finite constant. Note that $n$ here denotes the total sample size (e.g., of a combined experimental / observational sample).  Similarly, we assume that $\tuh, \tbh$ are asymptotically normal, in that 
  \begin{equation}
    \srn\begin{pmatrix}
      \tuh - \tz \\ \tbh - (\tz + \mu)
    \end{pmatrix} \cid \cN\left(0, \begin{bmatrix}
    \var(\infu) & \cov(\infu, \infb) \\
    \cov(\infu, \infb) & \var(\infu)
    \end{bmatrix}\right)
  \end{equation}
\end{assumption}

This assumption requires that both $\tuh$ and $\tbh$ are asymptotically normal for their respective limits $\tz + \mu$ and $\tz$, a condition that is broadly satisfied by commonly used estimators. This formulation also suggests a natural choice for estimating the variance of $\tuh$, the variance of $\tbh$, and the covariance, by first estimating 
\begin{align*}
  \hat{\var}(\infu) &= \frac{1}{n} \sum^{n}_{i=1} {(\infuh(Z_i))}^2 & \hat{\var}(\infb) &= \frac{1}{n} \sum^{n}_{i=1} {(\infbh(Z_i))}^2 & \hat{\cov}(\infu, \infb) &= \frac{1}{n} \sum^{n}_{i=1} \infuh(Z_i) \infbh(Z_i),
\end{align*}
where $\infuh, \infbh$ are consistent plug-in estimators $\infuh \cip \infu, \infbh \cip \infb$ for their respective influence functions.  The components of $\lambdah$ are then given by 
\begin{align*}
  \sigmauhs &= n^{-1} \hat{\var}(\infu) & \sigmabhs &= n^{-1} \hat{\var}(\infb) & \sigmacovh &= n^{-1} \hat{\cov}(\infu, \infb)
\end{align*}
and the conditions of~\cref{thm:consistency} are satisfied, as $\hat{\var}(\infu), \hat{\var}(\infb), \sigmacovh$ all converge to a constant by the weak law of large numbers.  They are also consistent in this case, but that is not required for~\cref{thm:consistency} to hold.

\subsection{Proofs of Main Results}%
\label{sec:proofs_of_main_results}

\Consistency*
\begin{proof}
  To see that $\lambdah \cip 0$ when $\mu \neq 0$, we can observe that 
  \begin{align*}
    \lambdah &= \frac{n\sigmauhs - n\sigmacovh}{n{(\tuh - \tbh)}^2 + n\sigmauhs + n\sigmabhs -2n\sigmacovh},
  \end{align*}
  which converges to zero by the fact that $\tuh - \tbh \cip \mu$, $n \sigmauhs \cip \nu_u$, $n \sigmabhs \cip \nu_b$, and $n \sigmacovh \cip \nu_{bu}$.  The denominator diverges to infinity due to the extra factor of $n$ in the denominator, i.e., $n{(\tuh - \tbh)}^2 \cip \infty$, and the result follows from the continuous mapping theorem. The fact that $\lambdah \cip 0$ when $\mu \neq 0$ is sufficient to conclude that $\thlh \cip \tz$ by another application of the continuous mapping theorem to the expression $\thlh = \tuh + \lambdah(\tbh - \tuh)$, and the fact that $\tbh - \tuh \cip \mu$.
\end{proof}

\newcommand{\m}{{\color{RawSienna}\delta}}
\newcommand{\ms}{{\color{RawSienna}\delta^{\ast}}}

\MSEBoundCovUnknownVar*
\begin{proof}
To derive the upper bound, we consider a perturbed version $\thlh(\m)$ of $\thlh$ that is parameterized by $\m \in \R$, where $\thlh(0)$ corresponds to the unperturbed $\thlh$, and where $\thlh(\m)$ is defined as 
\begin{align}
  \thlh(\m) &\coloneqq \lambdah(\m) (\tbh + \m) + (1 - \lambdah(\m)) \tuh & \lambdah(\m) &\coloneqq \frac{\sigmauhs - \sigmacovh}{{(\m + \tbh - \tuh)}^2 + \sigmauhs + \sigmabhs - 2 \sigmacovh}
\end{align}
This perturbation is equivalent to replacing $\tbh$ with $\tbh + \m$ as an input to the reference estimator $\thlh$, without changing any of the other inputs.  Because $\thlh(0) = \thlh$, we have it that
\begin{equation*}
    {(\thlh - \tz)}^2 \leq {\displaystyle \sup_{\m}} {(\thlh(\m) - \tz)}^2,
\end{equation*}
and so by the monotonicity of expectations,
  \begin{align}
    \E[{(\thlh - \tz)}^2] &\leq \E\left[ \sup_{\m} {(\thetahlh(\m) - \tz)}^2 \right]\nonumber\\
    &= \E\left[ \sup_{\m} {(\tuh - \tz + \lambdah(\m)(\m + \tbh - \tuh))}^2 \right]\nonumber\\
    &= \E\left[ \sup_{\m} \left({(\tuh - \tz)}^2 + {(\lambdah(\m)(\m + \tbh - \tuh))}^2 + 2(\tuh - \tz)(\lambdah(\m)(\m + \tbh - \tuh)) \right)\right]\nonumber\\
    &= \E\left[{(\tuh - \tz)}^2\right] + \E\left[ \sup_{\m} \left({(\lambdah(\m)(\m + \tbh - \tuh))}^2 + 2(\tuh - \tz)(\lambdah(\m)(\m + \tbh - \tuh))\right) \right]\nonumber\\
    &= \sigmaus +  \E\left[ \sup_{\m} \left({(\lambdah(\m)(\m + \tbh - \tuh))}^2 + 2(\tuh - \tz)\lambdah(\m)(\m + \tbh - \tuh)\right) \right].\label{eq:sup_mse_cov}
\end{align}
Our strategy is then to calculate the supremum inside of the expectation, and then give the bound in terms of the remaining parameters. 

\begin{restatable}{lemma}{MSEBoundLemma}\label{lemma:MSEBoundCov_supremum}
  The optimizer $\ms$ that achieves the supremum in Equation~\eqref{eq:sup_mse_cov} is given by 
  \begin{equation}\label{eq:optimal_delta_mse_bound}
    \ms = \begin{cases}
      (\tuh - \tbh) + \sqrt{\sigmauhs + \sigmabhs - 2 \sigmacovh}, &\ \text{if } (\sigmauhs - \sigmacovh) (\tuh - \tz) \geq 0\\
      (\tuh - \tbh) - \sqrt{\sigmauhs + \sigmabhs - 2 \sigmacovh}, &\ \text{if } (\sigmauhs - \sigmacovh) (\tuh - \tz) < 0.
    \end{cases}
  \end{equation}
  and the associated optimal value is
  \begin{align}
  {(\lambdah(\ms)(\ms + \tbh - \tuh))}^2 + 2(\tuh - \tz)\lambdah(\ms)(\ms + \tbh - \tuh) &= \frac{\abs{\sigmauhs - \sigmacovh}\abs{\tuh - \tz}}{\sqrt{\sigmauhs + \sigmabhs - 2 \sigmacovh}} + \frac{(\sigmauhs - \sigmacovh)}{4 (\sigmauhs + \sigmabhs - 2 \sigmacovh)}
  \end{align}
\end{restatable}
We defer proof of~\cref{lemma:MSEBoundCov_supremum} to~\cref{sec:proof_of_lemma}. Using~\cref{lemma:MSEBoundCov_supremum}, we can write Equation~\eqref{eq:sup_mse_cov} as 
  \begin{equation}
    \E[{(\thetahlh - \tz)}^2] \leq \sigmaus + \E\left[\frac{\abs{\sigmauhs - \sigmacovh} \abs{\tuh - \tz}}{\sqrt{\sigmauhs + \sigmabhs - 2 \sigmacovh}} + \frac{{(\sigmauhs - \sigmacovh)}^2}{4 {(\sigmauhs + \sigmabhs - 2 \sigmacovh)}}\right]
  \end{equation}
  which can be refined by a change of notation, writing $S \coloneqq \frac{\abs{\sigmauhs - \sigmacovh}}{\sqrt{\sigmauhs + \sigmabhs - 2 \sigmacovh}}$, and observing that this yields
  \begin{align*}
    \E[{(\thetahlh - \tz)}^2] &\leq \sigmaus + \E\left[S \abs{\tuh - \tz} + \frac{S^2}{4}\right] \\
                              &=  \sigmaus + \E\left[S \abs{\tuh - \tz}\right] + \frac{1}{4}\E[S^2] \\
                              &\leq  \sigmaus + \sqrt{\E[S^2]\E\left[{(\tuh - \tz)}^2\right]} + \frac{1}{4}\E[S^2] & \text{(Cauchy-Schwarz)} \\
                              &= \sigmaus + \sqrt{\E[S^2]} \sigmau + \frac{1}{4} \E[S^2]\\
                              &= {\left(\sigmau + \frac{1}{2} \sqrt{\E[S^2]}\right)}^2
  \end{align*}
  where $\E[S^2] = \E[{(\sigmauhs - \sigmacovh)}^2 / (\sigmauhs + \sigmabhs - 2 \sigmacovh)]$. 
\end{proof}

\MSEBoundCov*
\begin{proof}
  To arrive at the desired form, we observe that $S^2$ (from~\cref{thmthm:mse_bound_cov_unknown_var}) becomes deterministic with known variances and covariances, i.e., $\sqrt{\E[S^2]} = \abs{\sigmaus - \sigmacov} / \sqrt{(\sigmaus + \sigmabs - 2 \sigmacov)}$. This allows us to write the bound from~\cref{thmthm:mse_bound_cov_unknown_var} as
  \begin{equation}\label{eq:mse_bound_cov_before_reparameterize}
    \E[{(\thetahlh - \tz)}^2]  \leq {\left(\sigmau + \frac{1}{2} \frac{\abs{\sigmaus - \sigmacov}}{\sqrt{(\sigmaus + \sigmabs - 2 \sigmacov)}}\right)}^2 
  \end{equation}
Define the correlation as $\rho \coloneqq \sigmacov / \sqrt{\sigmaus \sigmabs}$, and define the ratio of the standard deviations as $c = \sqrt{\sigmabs/\sigmaus}$.  If $\sigmabs = 0$, we define $\rho = 0$.  Note that by Assumption~\ref{asmp:nonzero-variance-diff}, $\sigmaus > 0$, so $c$ is always well-defined. 
\begin{align*}
  \sigmacov &= \sigmaus \rho c &\implies && \sigmaus - \sigmacov &= \sigmaus (1 - \rho c) \\
  \sigmabs &= \sigmaus c^2 & \implies && \sigmaus + \sigmabs - 2 \sigmacov &= \sigmaus(1 + c^2 - 2 \rho c)
\end{align*}
and observe that this allows us to rewrite Equation~\eqref{eq:mse_bound_cov_before_reparameterize} as 
\begin{align}
  \E[{(\thetahlh - \tz)}^2] &\leq \sigmaus + \frac{\sigmaus \abs{1 - \rho c}\sigmau}{\sqrt{\sigmaus(1 - 2 \rho c + c^2)}} + \frac{{(\sigmaus(1 - \rho c))}^2}{4\sigmaus(1 - 2 \rho c + c^2) } \nonumber \\
  &= \sigmaus + \sigmaus \frac{\abs{1 - \rho c}}{\sqrt{1 - 2 \rho c + c^2}} + \sigmaus \frac{{(1 - \rho c)}^2}{4(1 - 2 \rho c + c^2) } \nonumber \\ 
  &= \sigmaus {\left(1 + \frac{1}{2} \frac{\abs{1 - \rho c}}{\sqrt{1 - 2 \rho c + c^2}}\right)}^2 \label{eq:mse_bound_cov_after_reparameterize}
\end{align}
which gives the desired result.
\end{proof}

\ConvergeUnboundedBias*
\begin{proof}
  First, we prove that the estimator $\thetahlh$ converges almost surely to the unbiased estimator $\tuh$ as $k \rightarrow \infty$
  \begin{equation}
   \thetahlhk \cas \tuh,
  \end{equation}
  which implies almost-sure convergence of the squared error
  \begin{equation}
    {\left(\thetahlhk - \tz\right)}^2 \cas {\left(\tuh - \tz\right)}^2.
  \end{equation}
  Once we have established this, we use the dominated convergence theorem to give the desired result.

  \textbf{Almost-Sure Convergence}: Let $\Omega$ denote the sample space, such that $\tuh(\omega), \tbh'(\omega)$ are the realized values of $\tuh, \tbh'$ for the event $\omega \in \Omega$.  For any realization of $\tuh(\omega), \tbh'(\omega)$, we will show that the estimator $\thetahlhk(\omega)$ converges to $\tuh(\omega)$ as $k \rightarrow \infty$.  In the sequel, we will drop the $\omega$ for simplicity of presentation, and consider any realization of $\tuh, \tbh'$, and write $\tbhk$ as the value $\tbh' + \mu_k$
  \begin{align*}
    \thetahlh &= \frac{(\tbh' + \mu_k) \cdot (\sigmaus - \sigmacov)}{{(\tuh - \tbh' - \mu_k)}^2 + \sigmaus + \sigmabs - 2 \sigmacov} + \frac{\tuh \cdot ({(\tuh - \tbh' - \mu_k)}^2 + \sigmabs - \sigmacov)}{{(\tuh - \tbh' - \mu_k)}^2 + \sigmaus + \sigmabs - 2 \sigmacov} ,
  \end{align*}
  where we can see that as $\mu_k \rightarrow \infty$, the first term goes to zero, given $\mu_k^2$ in the denominator and $\mu_k$ in the numerator.  Meanwhile, the second term converges to $\tuh$, as $\lambdah$ converges to 0, so that $(1 - \lambdah)\tuh$ converges to $\tuh$
  \begin{equation}
    \lim_{k \rightarrow \infty} \frac{{(\tuh - \tbh' - \mu_k)}^2 + \sigmabs - \sigmacov}{{(\tuh - \tbh' - \mu_k)}^2 + \sigmaus + \sigmabs - 2 \sigmacov} = 1
  \end{equation}

\textbf{Dominated Convergence Theorem:} 
Here we apply the dominated convergence theorem, by defining a random variable $Z$ such that ${(\thetahlhk - \tz)}^2 \leq Z$ almost surely, and where $\E[Z] < \infty$. First, we can observe that ${(\thetahlhk - \tz)}^2$ is a function of $\mu_k$, and is upper bounded by the supremum over all possible values of $\mu_k \in \R$.  In the following, we represent this by replacing $\mu_k$ with the value $\delta$. 
\begin{align*}
  {(\thetahlhk - \tz)}^2 &= {(\tuh - \tz)}^2 + 2 \lambdah (\tbh' + \mu_k - \tuh)(\tuh - \tz) + \lambdah^2 {(\tbh' + \mu_k - \tuh)}^2 \\
  &\leq \sup_{\delta \in \R} {(\tuh - \tz)}^2 + 2 \lambdah (\tbh' + \delta - \tuh)(\tuh - \tz) + \lambdah^2 {(\tbh' + \delta - \tuh)}^2
\end{align*}
The first term does not depend on $\delta$, and by~\Cref{lemma:MSEBoundCov_supremum}, the remainder is maximized by taking $\delta^*$ as defined in Equation~\eqref{eq:optimal_delta_mse_bound}.  This yields that 
\begin{align*}
  {(\thetahlhk - \tz)}^2 &\leq {(\tuh - \tz)}^2 + \frac{\abs{\sigmaus - \sigmacov}\abs{\tuh - \tz}}{\sqrt{\sigmaus + \sigmabs - 2 \sigmacov}} + \frac{(\sigmaus - \sigmacov)}{4 (\sigmaus + \sigmabs - 2 \sigmacov)}
\end{align*}
the right-hand side is a random variable that does not depend on $\mu_k$, and it has a finite expectation, as shown in~\Cref{thmcorr:mse_bound_cov}.  This completes the proof.
\end{proof}

\subsection{Proof of~\texorpdfstring{\cref{lemma:MSEBoundCov_supremum}}{Lemma}}%
\label{sec:proof_of_lemma}

\newcommand{\bterm}{{\color{RawSienna}\Delta}}
\newcommand{\bterms}{{\color{RawSienna}\Delta^2}}
\newcommand{\stermu}{S_u}
\newcommand{\stermb}{S_b}

\MSEBoundLemma*
\begin{proof}
Because the expression in~\cref{eq:sup_mse_cov} is a differentiable function of $\m$, we enumerate all of the stationary points, and demonstrate that the chosen value achieves the maximum objective value over all such stationary points. We will use the simplifying expressions
\begin{align*}
  \bterm &\coloneqq \m + \tbh - \tuh & \stermu &\coloneqq \sigmauhs - \sigmacovh & \stermb &\coloneqq \sigmabhs - \sigmacovh 
\end{align*}
which allows us to rewrite $\lambdah(\m)$ as follows, recalling that $\bterm$ is a function of $\m$
\begin{align*}
  \lambdah(\m) &= \frac{\sigmauhs - \sigmacovh}{{(\m + \tbh - \tuh)}^2 + \sigmauhs + \sigmabhs - 2 \sigmacovh} = \frac{\stermu}{\bterms + \stermu + \stermb},
\end{align*}
and we will expand the supremum in~\cref{eq:sup_mse_cov} to write it as a function of $\bterm$ for notational convenience, noting that $\bterm$ is a 1-to-1 function of $\m$ for a given $\tuh, \tbh$
\begin{align}
  \sup_{\m} 2 \lambdah(\m) (\m + \tbh - \tuh)(\tuh - \tz) + \lambdah{(\m)}^2 {(\m + \tbh - \tuh)}^2 &= \sup_{\bterm} 2 \lambdah(\m) \bterm (\tuh - \tz) + \lambdah{(\m)}^2 \bterms \\
  &= \sup_{\bterm} 2 \frac{\stermu \bterm (\tuh - \tz)}{\bterms + \stermu + \stermb} + 
  \frac{\stermu^2\bterms}{{(\bterms + \stermu + \stermb)}^2} \\
  &= \sup_{\bterm} \frac{2 (\bterms + \stermu + \stermb)\stermu \bterm (\tuh - \tz) +\stermu^2\bterms}{{(\bterms + \stermu + \stermb)}^2} \label{eq:sup_mse_cov_shorthand}
\end{align}
Before enumerating the stationary points in this expression, we demonstrate that the maximum is attained by a finite value of $\bterm$.  In particular, as $\abs{\bterm} \rightarrow \infty$, the entire term goes to zero, as the denominator is $O(\bterm^4)$ while the numerator is $O(\bterm^3)$. This justifies the use of the first-order condition to identify local maxima and minima, observing that $(d/d\m) \bterm= 1$ and $(d / d\m) S = 0$, and that taking the supremum with respect to $\bterm$ is equivalent to doing so with respect to $\m$, since $\bterm$ is simply $\m$ plus a fixed offset.

First, we compute the derivative of the numerator and denominator of the expression in Equation~\eqref{eq:sup_mse_cov_shorthand}
\begin{align}
 \frac{d}{d\bterm} \left[2 (\bterms + \stermu + \stermb)\stermu \bterm (\tuh - \tz) +\stermu^2 \bterms\right] &= \frac{d}{d\bterm} \left[2 (\stermu \bterm^3 + \stermu^2 \bterm + \stermu \bterm \stermb)(\tuh - \tz) +\stermu^2\bterms\right]\\
 &= (6 \stermu \bterms + 2 \stermu^2 + 2 \stermu \stermb)(\tuh - \tz) + 2 \stermu^2\bterm\\
  \frac{d}{d\bterm} {(\bterms + \stermu + \stermb)}^2 &= 4 {(\bterms + \stermu + \stermb)} \bterm
\end{align}
and then we compute the derivative of the expression in Equation~\eqref{eq:sup_mse_cov_shorthand}
\begin{align}
  \frac{d}{d\bterm} \frac{2 (\bterms + \stermu + \stermb)\stermu \bterm (\tuh - \tz) +\stermu^2\bterms}{{(\bterms + \stermu + \stermb)}^2} =& \frac{{(\bterms + \stermu + \stermb)}^2 \cdot [(6 \stermu \bterms + 2\stermu^2 + 2\stermu \stermb)(\tuh - \tz) + 2 \stermu^2\bterm ]}{{(\bterms + \stermu + \stermb)}^4} \nonumber \\
  &\quad \quad - \frac{[ 2 (\bterms + \stermu + \stermb)\stermu \bterm (\tuh - \tz) +\stermu^2\bterms ]4 {(\bterms + \stermu + \stermb)}  \bterm}{{(\bterms + \stermu + \stermb)}^4},
\end{align}
and simplify this expression to find values of $\bterm$ for which this is equal to zero. Note that by Assumption~\ref{asmp:nonzero-variance-diff}, $\stermu + \stermb > 0$, so that the denominator term is non-zero, and we can safely remove it (and a similar term in the numerator) by multiplying the entire expression by ${(\bterms +\stermu + \stermb)}^3$.
\begin{align}
  0 &= {(\bterms + \stermu + \stermb)} \cdot [(6 \stermu \bterms + 2\stermu^2 + 2\stermu \stermb)(\tuh - \tz) + 2 \stermu^2\bterm ] \nonumber \\
       &- [ 2 (\bterms + \stermu + \stermb)\stermu \bterm (\tuh - \tz) + \stermu^2\bterms ]4 \bterm \nonumber \\ 
  \iff & {(\bterms + \stermu + \stermb)} (3 \stermu \bterms + \stermu^2 + \stermu \stermb)(\tuh - \tz) + {(\bterms + \stermu + \stermb)} \stermu^2\bterm \label{eq:divide_by_two} \\
       &-  4 (\bterms + \stermu + \stermb)\stermu \bterms (\tuh - \tz) - 2\stermu^2{\bterm}^3  = 0  \nonumber \\
  \iff & {(\bterms + \stermu + \stermb)} (\stermu + \stermb - \bterms) \stermu (\tuh - \tz) + {(\bterms + \stermu + \stermb)} \stermu^2\bterm - 2\stermu^2{\bterm}^3  = 0   \label{eq:collect_terms}\\
  \iff & (\stermu + \stermb - \bterms) \stermu (\tuh - \tz) + \stermu^2\bterm = \frac{2\stermu^2{\bterm}^3}{\bterms + \stermu + \stermb} \nonumber \\
  \iff & (\stermu + \stermb - \bterms) \stermu (\tuh - \tz) = \frac{2\stermu^2{\bterm}^3 - (\stermu^2 \bterm)(\bterms + \stermu + \stermb)}{\bterms + \stermu + \stermb}\nonumber \\
  \iff & (\stermu + \stermb - \bterms) \stermu (\tuh - \tz) = \frac{\stermu^2{\bterm}^3 - \stermu^3 \bterm - \stermu^2 \bterm \stermb}{\bterms + \stermu + \stermb}\nonumber \\
  \iff & (\stermu + \stermb - \bterms) \stermu (\tuh - \tz) = - \stermu^2 \bterm \frac{-\bterms + \stermu + \stermb}{\bterms + \stermu + \stermb}   \label{eq:bterm2solution} \\
  \iff & \stermu (\tuh - \tz) = \frac{- \stermu^2 \bterm}{\bterms + \stermu + \stermb} & \text{(If $\bterms \neq \stermu + \stermb$)}  \nonumber \\
  \iff & (\bterms + \stermu + \stermb) \stermu (\tuh - \tz) + \stermu^2 \bterm = 0 \label{eq:final_stationary_point}
\end{align}
In Equation~\eqref{eq:divide_by_two} we divide by two and distribute terms, and in Equation~\eqref{eq:collect_terms} we collect terms involving $\tuh - \tz$ before simplifying further. Equation~\eqref{eq:bterm2solution} reveals that $\bterms = \stermu + \stermb$ is a stationary point, and Equation~\eqref{eq:final_stationary_point} implicitly defines another set of stationary points. Any stationary point is a solution to one of the following.
\begin{align}
  (\bterms + \stermu + \stermb) \stermu (\tuh - \tz) + \stermu^2 \bterm &= 0 \label{eq:first_order_condition}\\
  \bterms &= \stermu + \stermb \label{eq:first_order_condition_main} 
\end{align}
Next, we will show that for finding a global maximum, it suffices to consider Equation~\eqref{eq:first_order_condition_main}.  In particular, we demonstrate that when we plug these conditions into the original expression from Equation~\eqref{eq:sup_mse_cov_shorthand},
\begin{equation*}
  \frac{2 (\bterms + \stermu + \stermb)\stermu \bterm (\tuh - \tz) +\stermu^2\bterms}{{(\bterms + \stermu + \stermb)}^2} = \begin{cases}
    \leq 0 &\ \text{if } \bterm \text{ satisfies Eq.~\eqref{eq:first_order_condition}}\\
    \frac{\pm  4 \stermu (\sqrt{\stermu + \stermb})(\tuh - \tz) + \stermu^2}{4 {(\stermu + \stermb)}}, &\ \text{if } \bterm \text{ satisfies Eq.~\eqref{eq:first_order_condition_main}} \\
  \end{cases}
\end{equation*}
where the stationary points satisfying Equation~\eqref{eq:first_order_condition_main} always include a non-negative solution, while those satisfying Equation~\eqref{eq:first_order_condition} are always non-positive.  We prove both of these points below.

\textit{Solutions satisfying Equation~\eqref{eq:first_order_condition}}: The solutions implied by Equation~\eqref{eq:first_order_condition} satisfy
\begin{equation}
(\bterms + \stermu + \stermb) \stermu (\tuh - \tz) + \stermu^2 \bterm = 0
\end{equation}
which includes the solution $\bterm = 0$ when $\tuh = \tz$, in which case the value of the optimization objective is zero.  When $\tuh \neq \tz$, we have it that $\bterm = 0$ is no longer a solution.  We can, however, use Equation~\eqref{eq:first_order_condition} to observe that $(\bterms + \stermu + \stermb) \stermu (\tuh - \tz) = - \stermu^2 \bterm$, which implies that the value of the objective is given by
\begin{equation}
  \frac{2 (\bterms + \stermu + \stermb)\stermu \bterm (\tuh - \tz) +\stermu^2\bterms}{{(\bterms + \stermu + \stermb)}^2} = \frac{2 \bterm (-\stermu^2 \bterm) + \stermu^2 \bterms}{{(\bterms + \stermu + \stermb)}^2} 
  = \frac{- \stermu^2 \bterms}{{(\bterms + \stermu + \stermb)}^2} \leq 0
\end{equation}
which is equal to zero if and only if $\stermu = 0$, in which case $\lambdah$ is zero.  Because this expression is non-positive, it is not a global maximum, since there exist solutions that are positive (see below).

\textit{Solutions satisfying Equation~\eqref{eq:first_order_condition_main}:} When $\bterms = \stermu + \stermb$, we see that Equation~\eqref{eq:sup_mse_cov_shorthand} is equal to
\begin{align*}
\frac{\pm  4(\stermu + \stermb) \stermu (\sqrt{\stermu + \stermb})(\tuh - \tz) + \stermu^2 {(\stermu + \stermb)}}{4 {(\stermu + \stermb)}^2} &= \frac{\pm  4 \stermu (\sqrt{\stermu + \stermb})(\tuh - \tz) + \stermu^2}{4 {(\stermu + \stermb)}}
\end{align*}
which is maximized by taking the absolute value of the first term, choosing $\bterm = \sqrt{\stermu + \stermb}$ if $\stermu (\tuh - \tz) > 0$ and $\bterm = - \sqrt{\stermu + \stermb}$ if $\stermu (\tuh - \tz) < 0$.  If $\tuh = \tz$ the choice of sign is irrelevant, and can be chosen arbitrarily.  This will always yield a non-negative solution, given by 
\begin{equation}
\frac{4 \abs{\stermu} (\sqrt{\stermu + \stermb})\abs{\tuh - \tz} + \stermu^2}{4 {(\stermu + \stermb)}}
\end{equation}
which yields the claimed result, that the supremum is given by 
\begin{equation}\label{eq:mse_bound_cov_proof-replace-with-estimates}
\frac{\abs{\sigmauhs - \sigmacovh} \abs{\tuh - \tz}}{\sqrt{\sigmauhs + \sigmabhs - 2 \sigmacovh}} + \frac{{(\sigmauhs - \sigmacovh)}^2}{4 {(\sigmauhs + \sigmabhs - 2 \sigmacovh)}},
\end{equation}
which completes the proof.
\end{proof}

\section{Comparison to \texorpdfstring{\citet{Cheng2021-sn}}{Cheng et al. 2021}}%
\label{sec:comparison_to_cheng_et_al_2021}

\citet{Cheng2021-sn} take a similar approach to taking an adaptive linear combination of observational and experimental estimators.  We focus on a broader class of estimation problems that involve estimating a real-valued parameter, including the use of surrogate outcomes (as described in~\cref{sec:motivation_and_setup}), while they focus on CATE estimation with kernel regression, in the context of combining experimental and trial estimators.  However, while they approach a specific problem setting, their approach can be seen as a variant of the reference estimator on a conceptual level, and is nearly equivalent in the setting of ATE estimation via combination of experimental and observational data, with the exception of an additional hyperparameter that they introduce, scaling the estimated bias by a factor of $n^{-\beta}$.

Here, we give the approach of that work, in the context of ATE estimation, combining observational and experimental data.  We focus on the ATE in the trial population, which they denote as $\tau_0(\bv) = \E[Y_1 - Y_0 \mid \bV = \bv, Z = 0]$, where $Z = 0$ denotes the trial population and $Z = 1$ denotes the observational population, and $\bV$ denotes a set of covariates.  We use the notation $\tau$ instead of $\tau(\bv)$ because there is no conditioning set for the ATE, and we will take $\tau_0$ to be the target of inference. In this case, Equation 8 of \citet{Cheng2021-sn} becomes
\begin{equation}
  \tauh = \tauh^r + \eta (\tauh^o - \tauh_r)
\end{equation}
and the goal is to estimate the optimal value of $\eta$ from data.  Here, \citet{Cheng2021-sn} consider standard doubly-robust (DR) pseudo-outcomes for the treatment effect (see Equation 10 of \citet{Cheng2021-sn}), which are denoted as $\hat{\Psi}^r$ for the outcomes based on the trial data, and $\hat{\Psi}^{o}$ for outcomes based on the observational data.  Here we use $\hat{\Psi}$ to denote the pseudo-outcome when we use a plug-in estimate of nuisance parameters, and $\bar{\Psi}$ to denote the pseudo-outcome when we plug in the true values of the nuisance parameters.

\citet{Cheng2021-sn} consider locally constant kernel regression for estimation of CATE, which is unnecessary for ATE\@.  As a result, under the simplifying assumption that the distribution of $X$ is the same across the trial and observational study,\footnote{This avoids the need for the weights $\omega(X)$ in their equations, but this is only for the sake of notational simplicity here.}
\begin{equation}
  \tauh^r = n^{-1} \sum^{n}_{i=1} \frac{\1{Z_i = 0}}{\hat{P}(Z = 0)} \hat{\Psi}^r_j
\end{equation}
where $\hat{P}(Z = 0) = n^{-1} \sum^{n}_{i=1} \1{Z_i = 0}$ is the empirical estimate of the proportion of the total dataset in the trial, with an analogous estimator for the observational data (see Equations 11--12 of \citet{Cheng2021-sn}).  Here, it is assumed throughout that this probability is bounded away from zero, so this work excludes the case where the number of observational samples is of a different asymptotic order than the number of trial samples.

With this in mind, both estimators can be written with an asymptotically linear representation as follows, where $\bar{\tau}$ is used as the asymptotic limit of an estimator $\hat{\tau}$
\begin{equation}
  \srn(\tauh^r - \bar{\tau}^r) = n^{-1/2} \sum^{n}_{i=1} \underbrace{\frac{\1{Z_i = 0}}{P(Z = 0)} (\bar{\Psi}^r_j - \bar{\tau}^r)}_{\xi_i^r} + o_p(1),
\end{equation}
and likewise for $\tauh^o$, replacing the superscript $r$ with $o$, and $Z = 0$ with $Z = 1$.  Here, the term $\xi_i^r$ is the influence function.  Per Lemma 5, the MSE for the target parameter $\bar{\tau}^r$ can be written as 
\begin{equation}
  \E\left[{\left(\tauh - \bar{\tau}^r\right)}^2\right] = n^{-1} \E\left[{\left(\xi_i^r - \eta(\xi^r_i - \xi_i^o)\right)}^2\right] + \eta^2{\left(\bar{\tau}^o - \bar{\tau}^r\right)}^2 + o(n^{-1/2}),
\end{equation}
which suggests the following scaled empirical criterion (multiplying by $n^2$), for estimating $\eta$, 
\begin{equation}
  \hat{Q}(\eta) = \sum^{n}_{i=1} {\left[\hat{\xi}^r_i - \eta (\hat{\xi}^r_i - \hat{\xi}^o_i)\right]}^2 + \eta^2 n^{(2 - \beta)} {(\tauh^o - \tauh^r)}^2
\end{equation}
where $\hat{\eta} = \argmin_{\eta \in \R} \hat{Q}(\eta)$, and where the empirical influence functions are estimated via plug-in, where e.g., $\hat{\xi}^r_i = \frac{\1{Z_i = 0}}{P(Z = 0)} (\hat{\Psi}^r_i - \tauh^r)$.  Here, the solution is given by the following
\begin{align*}
\eta & = \frac{\E_n\left[{(\hat{\xi}^r)}^2\right] - \E_n[\hat{\xi}^r \hat{\xi}^o] }{\E_n[{(\hat{\xi}^r - \hat{\xi}^o)}^2] + n^{(1 - \beta)} {(\tauh^o - \tauh^r)}^2}
\end{align*}
where $\E_n[\cdot]$ is the empirical average.  

\textbf{Connection to the reference estimator}: Based on the asymptotically linear form of these estimators, we can write our estimators of the variance and covariance of each estimator as
\begin{align*}
  \sigmauhs &\coloneqq \frac{1}{n} \E_n\left[{(\hat{\xi}^r)}^2\right]& \sigmabhs &\coloneqq \frac{1}{n} \E_n\left[{(\hat{\xi}^o)}^2\right] & \sigmacovh &\coloneqq \frac{1}{n} \E_n\left[{\hat{\xi}^r \hat{\xi}^o}\right] 
\end{align*}
and observe that the proposed estimator of \citet{Cheng2021-sn}, adapted to the setting of ATE estimation, would be equal to a similar affine combination as the reference estimator, with weights
\begin{align*}
  \eta & = \frac{\sigmauhs - \sigmacovh}{n^{-\beta} {(\tauh^r - \tauh^o)}^2 + \sigmauhs + \sigmabhs - 2 \sigmacovh}
\end{align*}
which differs from the choice of $\lambdah$ used in the reference estimator due to the $n^{-\beta}$ term in the denominator.

\section{Additional Experimental Details}%
\label{sec:additional_baseline_details}

\subsection{Prior Approaches}%
\label{sec:baselines}

\textbf{Hypothesis Testing} We give the methodology of \citet{Yang2020-na} in full generality here, before discussing how it applies to our setting.  They suppose that there exists some score function $S_{\psi}(V)$, where $\psi$ is the parameter of interest and $V$ denotes observed data, where $\delta = 0$ corresponds to the observational data and $\delta = 1$ correspond to the randomized data.  Let there be $m$ samples in the randomized data, denoted $\cA$, and $n$ samples in the observational data, denoted $\cB$. Solving for $\psi$ requires solving the moment condition $\E[S_{\psi}(V)] = 0$. The simplest example of such a score function approach is estimation of the mean of $V$, where $S_{\psi}(V) = V - \psi$, and solving for $\psi$ is simply given by observing that $\E_n[V] = \psi$.  

The core approach is to construct a statistic for testing whether or not $S_{\psi}(V)$ has the same average value in unbiased randomized trial data, versus in the potentially biased observational (or \enquote{real world}) data. The first step in constructing their test statistic is to estimate the parameter from the randomized data, denoting this estimate as $\hat{\psi}_{rt} = \tuh$, and then evaluate the score on the real-world data, giving 
\begin{equation}
n^{-1/2} \sum_{i \in \cB} \hat{S}_{rw, \hat{\psi}_{rt}}(V_i) = \sqrt{n}
({\tuh} - {\tbh}) 
\end{equation}
which is then used to construct the test statistic (see Equation 7 of \citet{Yang2020-na}) as 
\begin{equation}
  T_n = \frac{{({\tuh} - {\tbh})}^2}{n \sigmah}
\end{equation}
where $\sigmah$ is a consistent estimate of the asymptotic variance of $\sqrt{n} ({\tuh} - {\tbh})$.  In our setting, this asymptotic variance is given by an estimate of $\sigmaus + \sigmabh - 2 \sigmacov$.  This test statistic converges in distribution to a chi-square random variable under the null hypothesis that no bias exists. With this in mind, their estimator can be represented as follows
\begin{equation}
  \sum_{i \in \cA \cup \cB} \{ \delta_i \hat{S}_{\psi}(V_i) + \1{T_n <
      c_{\gamma}}(1 - \delta_i) \hat{S}_{\psi}(V_i) \} = 0
\end{equation}
where if $T_n \geq c_{\gamma}$, this reduces to using ${\tuh}$, and otherwise this reduces to pooling the data and taking a global average of $\tuh, \tbh$, weighted by sample size. In our experiments, $\tuh, \tbh$ have the same sample size, so this is just a simple average of $\tuh, \tbh$.

The asymptotic bias and MSE of this estimator (for a given threshold $c_{\gamma}$) depends on the underlying bias of the observational estimator. \citet{Yang2020-na} derive an analytical formula for these terms (see Corollary 1 of \citet{Yang2020-na}), and note that $c_{\gamma}$ can be tuned by first estimating the bias, plugging this into these formula, and choosing a threshold that minimizes the resulting MSE\@.  More practically, they suggest estimating the bias using a plug-in estimate, specifying a grid of values for the significance level $\gamma$, and simulating from the limiting mixture distribution to identify the significance level that minimizes the MSE\@. 

In our experiments, we implement this data-driven selection of the hyperparameter as follows: For each setting of parameters in Table~\ref{tab:simulation_parameter}, we simulate performance of this approach for a grid of significance levels $\gamma \in \{0, 0.05, 0.10, \ldots, 0.95, 1\}$.  For each value of the bias $\mu \in [0, 1.5]$, we record the threshold which yields minimum MSE\@.  Then, we re-run the simulations, where we first estimate the bias as $\abs{\tuh - \tbh}$, and then look up the optimal cutoff based on our prior simulations.

\textbf{Anchored Thresholding} Given an unbiased estimate $\tuh$ and a biased estimate $\tbh$, \citet{Chen2021-eo} always combine the estimators, but they first apply a bias correction to $\tbh$. In particular, they apply soft-thresholding to estimate the bias, where 
\begin{equation}
  \hat{\mu} = \begin{cases}
    \text{sign}({\tbh} - {\tuh})\left(\abs{{\tbh} - {\tuh}} - \lambda
    \sqrt{\hat{\var}({\tbh} - {\tuh})}\right), &\ \text{if } \abs{{\tbh} - {\tuh}}
    \geq \lambda \cdot \sqrt{\hat{\var}( {\tbh} - {\tuh} )} \\
    0, &\ \text{otherwise.}
  \end{cases}
\end{equation}
This estimated bias is used to \enquote{correct} $\tbh$ by replacing it with $\tbh - \hat{\mu}$.  At this stage, the estimators are combined on the assumption that both are unbiased, with the combination (in our setting) given by
\begin{equation}
  \hat{w} (\tbh - \hat{\mu}) + (1 - \hat{w}) \tuh
\end{equation}
where 
\begin{equation}
\hat{w} = \frac{\sigmauhs - \sigmacovh}{\sigmauhs + \sigmabhs - 2 \sigmacovh}.
\end{equation}
In this setting $\lambda$ is a hyperparameter, which should be of asymptotic order $\lambda \asymp \sqrt{\log n}$.  In their experiments, they choose a constant $\lambda_1 = 0.5$ and then set $\lambda = \lambda_1 \cdot \sqrt{\log n}$, so we do the same.

\section{Additional Experimental Results}%
\label{sec:additional_experimental_results}

\subsection{Comparison to hypothesis testing with a fixed threshold}%
\label{sec:comparison_to_hypothesis_testing}

In~\cref{fig:ht_app}, we compare against the testing-based approach of~\citet{Yang2020-na} for different \textbf{fixed} significance thresholds $\gamma$, in contrast to the \enquote{data-driven} thresholds used in the main text, and described in~\cref{sec:baselines}.

\begin{figure}[t]
  \centering
  \begin{subfigure}[t]{0.3\textwidth}
    \centering
    \includegraphics[width=\linewidth, height=\linewidth, keepaspectratio]{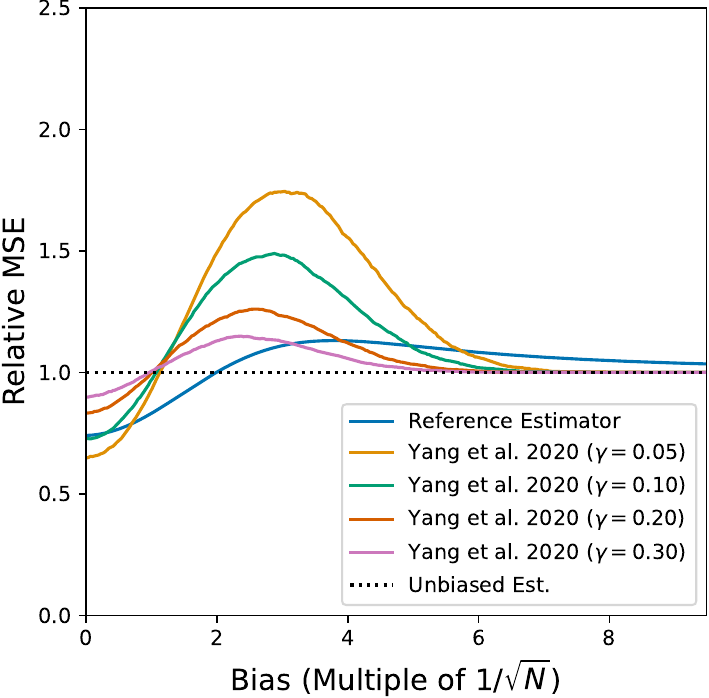}
    \caption{}%
    \label{fig:curves_ht_app}
  \end{subfigure}
  \begin{subfigure}[t]{0.3\textwidth}
    \centering
    \includegraphics[width=\linewidth, height=\linewidth, keepaspectratio]{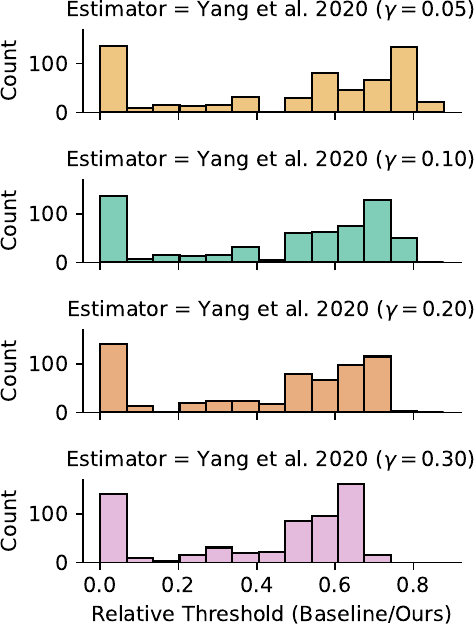}
    \caption{}%
    \label{fig:histogram_ht_app}
  \end{subfigure}
  \begin{subfigure}[t]{0.3\textwidth}
    \centering
    \includegraphics[width=\linewidth, height=\linewidth, keepaspectratio]{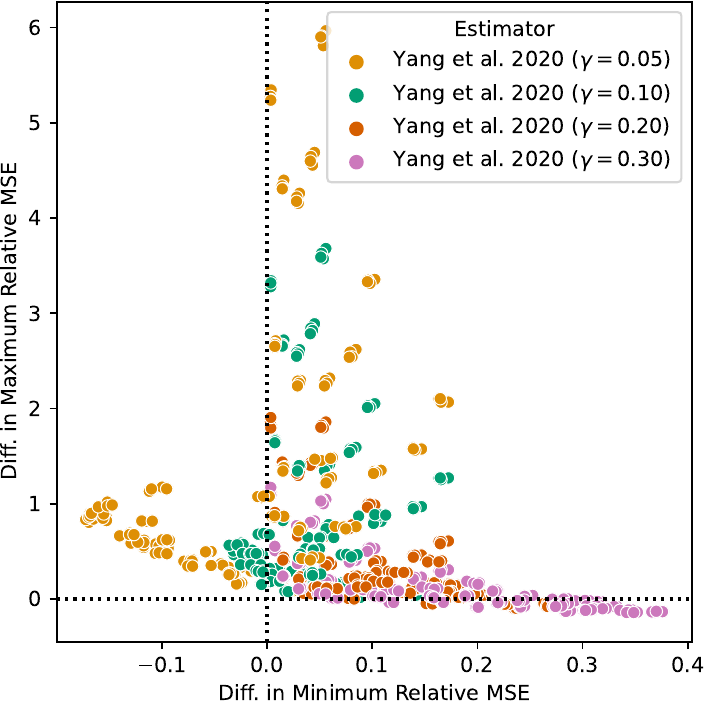}
    \caption{}%
    \label{fig:best_worst_2d_ht_app}
  \end{subfigure}
  \caption{Comparison against hypothesis testing approach of~\citet{Yang2020-na}.
      (\subref{fig:curves_ht_app}) See~\cref{fig:single_example_curve_baseline} for description.
      (\subref{fig:histogram_ht_app}) Histogram of the \textit{ratio} of bias thresholds (baseline / ours), where a number $<1$ indicates that the given approach has a lower threshold than the reference estimator, and where \enquote{ours} refers to the reference estimator.
      (\subref{fig:best_worst_2d_ht_app}) See~\cref{fig:best_worst_curve_diff} for description.
  }%
  \label{fig:ht_app}
\end{figure}

\subsection{Comparison of different hyperparameter settings}%
\label{sec:comparison_of_different_hyperparameter_settings}

In~\cref{fig:anchor_app} we compare against the anchored thresholding approach of~\citet{Chen2021-eo} for different choices of the hyperparameter $\lambda$.  Recall from~\cref{sec:prior_approaches_for_combining_estimators} that $\lambda$ controls the extent to which the estimated bias is regularized towards zero: Hence, large values of $\lambda$ are less conservative, and small values of $\lambda$ are more conservative.  The large variation in outcomes (best vs worst-case performance) speaks to the sensitivity of the method towards the choice of hyperparameter.  We give a similar comparison to~\citet{Cheng2021-sn} in~\cref{fig:l2_app}, whose approach is equivalent to the reference estimator for $\beta = 0$, and otherwise tends to be more optimistic, with a higher maximum and lower minimum relative MSE, along with a lower bias threshold.

\begin{figure}[t]
  \centering
  \begin{subfigure}[t]{0.3\textwidth}
    \centering
    \includegraphics[width=\linewidth, height=\linewidth, keepaspectratio]{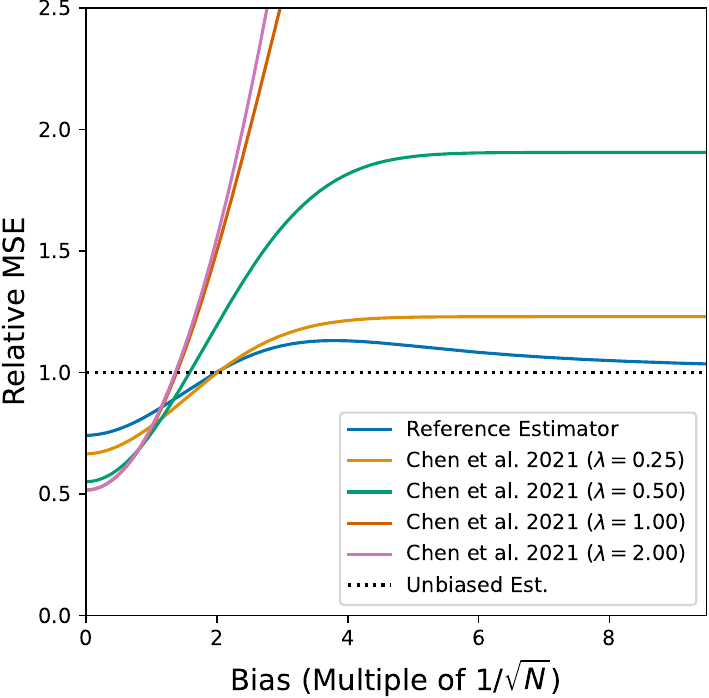}
    \caption{}%
    \label{fig:curves_anchor_app}
  \end{subfigure}
  \begin{subfigure}[t]{0.3\textwidth}
    \centering
    \includegraphics[width=\linewidth, height=\linewidth, keepaspectratio]{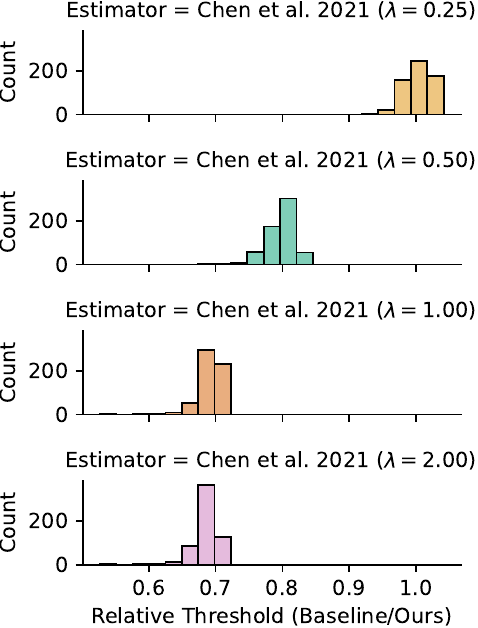}
    \caption{}%
    \label{fig:histogram_anchor_app}
  \end{subfigure}
  \begin{subfigure}[t]{0.3\textwidth}
    \centering
    \includegraphics[width=\linewidth, height=\linewidth, keepaspectratio]{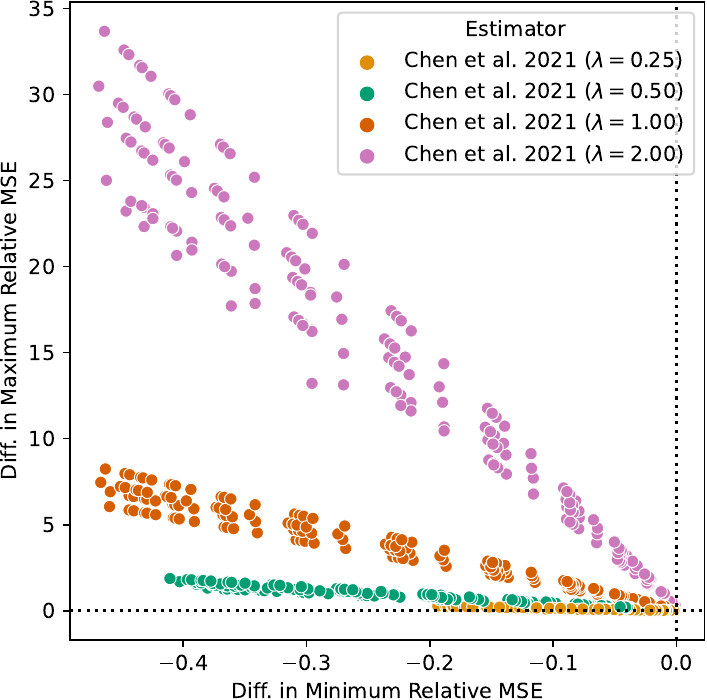}
    \caption{}%
    \label{fig:best_worst_2d_anchor_app}
  \end{subfigure}
  \caption{Comparison against anchored thresholding approach of~\citet{Chen2021-eo}.
      (\subref{fig:curves_anchor_app}) See~\cref{fig:single_example_curve_baseline} for description.
      (\subref{fig:histogram_anchor_app}) Histogram of the \textit{ratio} of bias thresholds (baseline / ours), where a number $<1$ indicates that the given approach has a lower threshold than the reference estimator, and where \enquote{ours} refers to the reference estimator.
      (\subref{fig:best_worst_2d_anchor_app}) See~\cref{fig:best_worst_curve_diff} for description.
  }%
  \label{fig:anchor_app}
\end{figure}

\begin{figure}[t]
  \centering
  \begin{subfigure}[t]{0.3\textwidth}
    \centering
    \includegraphics[width=\linewidth, height=\linewidth, keepaspectratio]{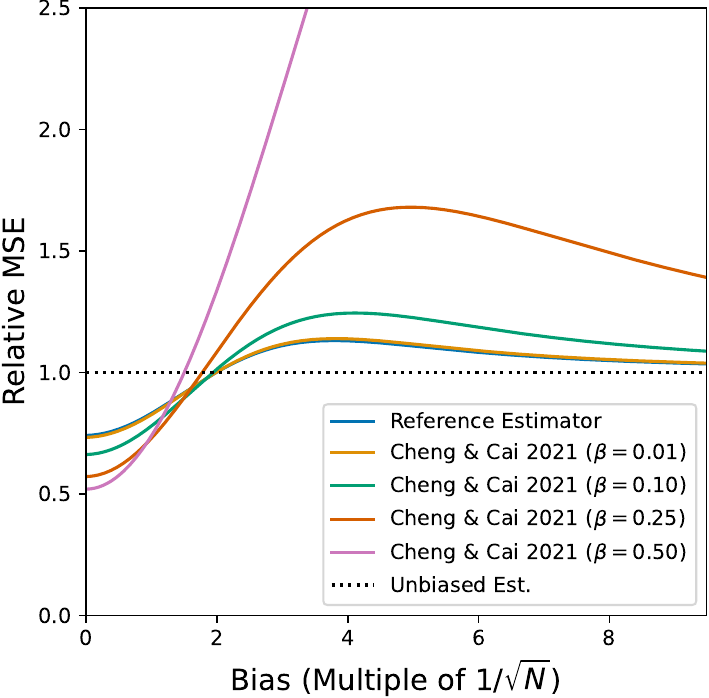}
    \caption{}%
    \label{fig:curves_l2_app}
  \end{subfigure}
  \begin{subfigure}[t]{0.3\textwidth}
    \centering
    \includegraphics[width=\linewidth, height=\linewidth, keepaspectratio]{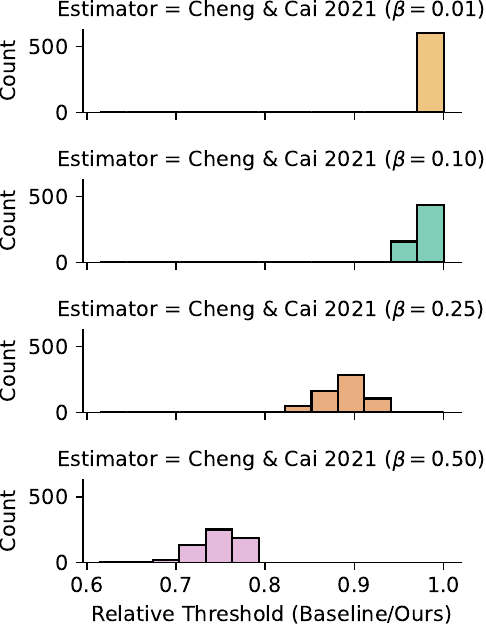}
    \caption{}%
    \label{fig:histogram_l2_app}
  \end{subfigure}
  \begin{subfigure}[t]{0.3\textwidth}
    \centering
    \includegraphics[width=\linewidth, height=\linewidth, keepaspectratio]{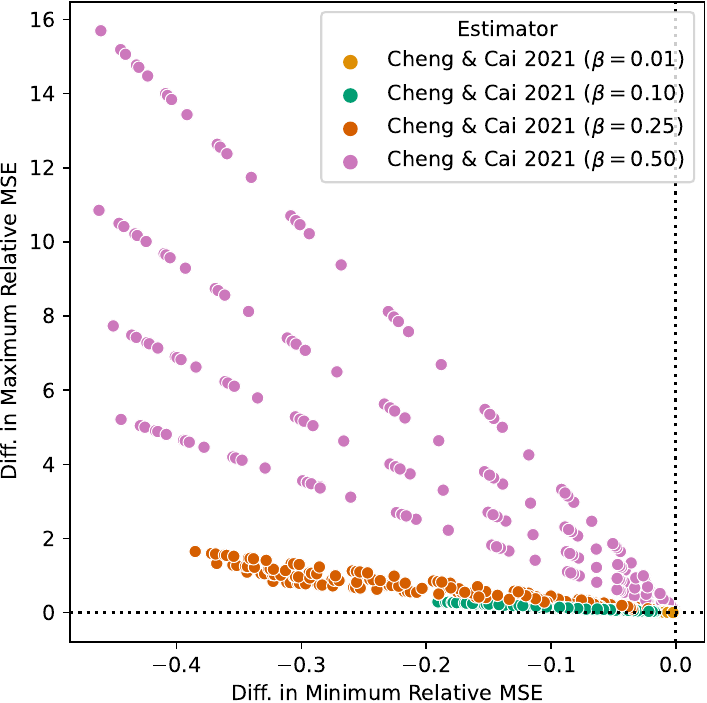}
    \caption{}%
    \label{fig:best_worst_2d_l2_app}
  \end{subfigure}
  \caption{{ Comparison against approach of~\citet{Cheng2021-sn}.
      (\subref{fig:curves_l2_app}) See~\cref{fig:single_example_curve_baseline} for description.
      (\subref{fig:histogram_l2_app}) Histogram of the \textit{ratio} of bias thresholds (baseline / ours), where a number $<1$ indicates that the given approach has a lower threshold than the reference estimator, and where \enquote{ours} refers to the reference estimator.
      (\subref{fig:best_worst_2d_l2_app}) See~\cref{fig:best_worst_curve_diff} for description.
    }
  }%
  \label{fig:l2_app}
\end{figure}

\subsection{Understanding factors that drive performance across settings}%
\label{sec:understanding_factors_that_drive_improvement_and_}

\begin{figure}[t]
\centering
  \begin{subfigure}[t]{0.3\textwidth}
  \centering
  \includegraphics[width=\linewidth]{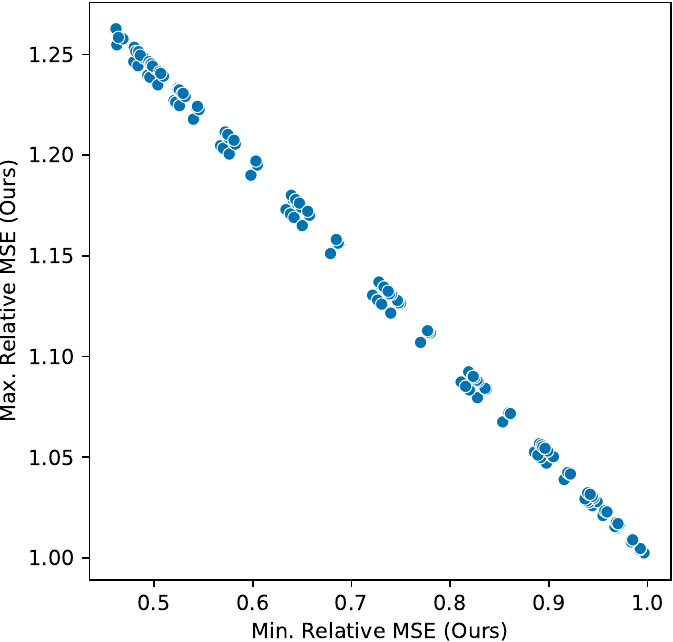}
  \caption{}%
  \label{fig:minmax_trend_all}
  \end{subfigure}
  \begin{subfigure}[t]{0.3\textwidth}
  \centering
   \includegraphics[width=\linewidth]{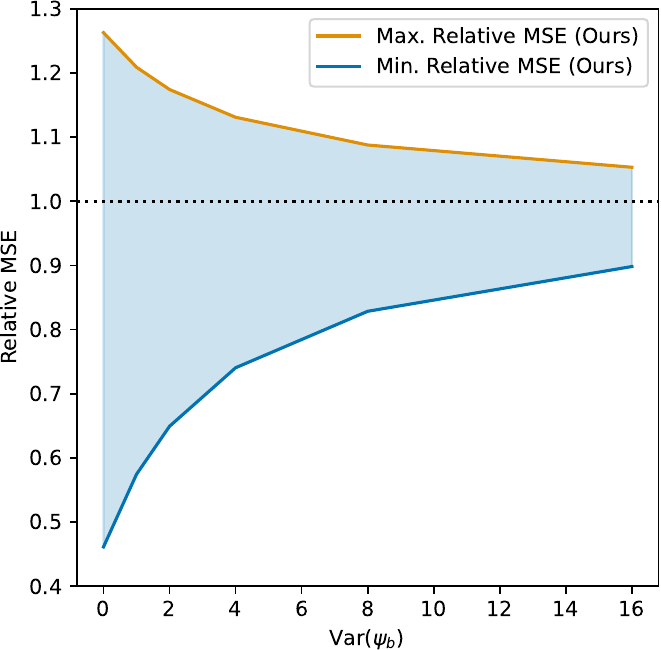}
   \caption{}%
  \label{fig:minmax_trend_indep}
  \end{subfigure}
  \begin{subfigure}[t]{0.3\textwidth}
  \centering
  \includegraphics[width=\linewidth]{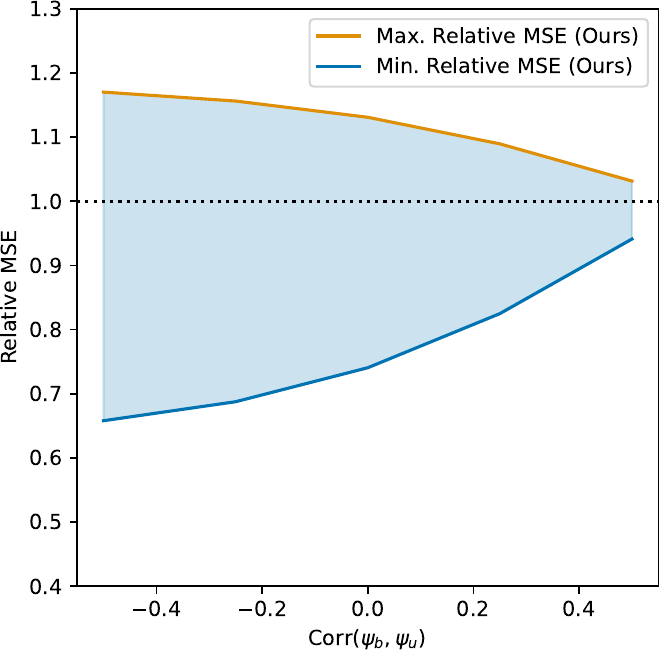}
  \caption{}%
  \label{fig:minmax_trend_corr}
  \end{subfigure}
  \caption{(\subref{fig:minmax_trend_all}) Every simulation setting (in terms of parameters in Table~\ref{tab:simulation_parameter}, excluding $\mu$) corresponds to a different dot, where on the x-axis we plot the best-case relative MSE (when $\mu = 0)$ and on the y-axis we plot the worst-case relative MSE (over all values of $\mu$, with other simulation parameters fixed).  Figures~(\subref{fig:minmax_trend_indep}-\subref{fig:minmax_trend_corr}) show variation in the minimum and maximum relative MSE of the reference estimator $\thlh$ where $n = 1000$, as a function of (\subref{fig:minmax_trend_indep}) variance in $\psib$, where $\var(\psiu) = 4, \text{corr}(\psib, \psiu) = 0$, and (\subref{fig:minmax_trend_corr}) correlation between $\psiu, \psib$ where $\var(\psiu) = 4, \var(\psib) = 4$.}%
  \label{fig:drivers}
\end{figure}

In this section, we take advantage of the wide range of simulation settings in~\cref{sec:empirical} to study how salient characteristics of the performance curves change for the reference estimator, as we vary the distribution $P$ that generates $\tuh, \tbh$.  We list some take-away observations here.

\textbf{The largest opportunities for improvement (e.g., low-variance $\tbh$) also have the highest worst-case error}: In Figure~\ref{fig:minmax_trend_all}, we examine the smallest and largest values of the relative MSE of $\thlh$ for each combination of simulation parameters, and observe that these values exhibit a nearly linear relationship: The larger the potential upside (when $\mu = 0$), the larger the potential downside (when $\mu$ is chosen adversarially). 
In Figures~\ref{fig:minmax_trend_indep} and~\ref{fig:minmax_trend_corr}, we demonstrate that the magnitude of the smallest/largest relative MSE depends on the relative benefit of incorporating $\tbh$:  In Figure~\ref{fig:minmax_trend_indep}, we show that both decrease in magnitude as the variance of $\tbh$ increases, and in Figures~\ref{fig:minmax_trend_corr}, we show that both decrease with increasing positive correlation of $\tbh$ and $\tuh$, and that both increase for more negative correlations, observations consistent with the worst-case bound given in~\Cref{thmcorr:mse_bound_cov} in~\cref{sec:theoretical}.

\textbf{The worst-case relative MSE of $\thlh$ is empirically bounded by a small constant factor:} Across all parameter settings, the largest relative MSE of $\thlh$ is bounded, never exceeding a 27\% increase in MSE over the use of $\tuh$ alone.  Moreover, measured by relative MSE, the potential upside is also larger than the potential downside, across all parameter settings. 

\textbf{How much bias is too much bias?} As shown in Figure~\ref{fig:single_example_curve_baseline}, the reference estimator $\thlh$ only improves upon $\tuh$ for sufficiently small values of the bias $\mu$.
Intuitively, we might expect the bias threshold to occur at some level where $\mu^2$ is of the same order as the variance of the difference $\var(\tuh - \tbh)$. With this intuition in mind, in Figure~\ref{fig:relativeMSE_biasVarianceRatio_pos_neg_corr}, we plot the relative MSE as a function of the ratio $\mu^2 / \var(\tuh - \tbh)$, and observe that when $\tuh, \tbh$ are independent (Figure~\ref{subfig:independent}), the maximum relative bias reliably falls around $\mu^2 \approx 2 \var(\tuh - \tbh)$.  This relationship also seems to hold when $\tuh, \tbh$ are correlated with similar variances.  However, when $\psib$ has a much lower variance than $\psiu$, then the maximum tolerable bias is lower (higher) when the two are positively (negatively) correlated (Figures~\ref{subfig:pos_corr}-\ref{subfig:neg_corr}).

\begin{figure*}[t]
\centering
  \begin{subfigure}[h]{0.3\linewidth}
  \centering
  \includegraphics[width=\linewidth]{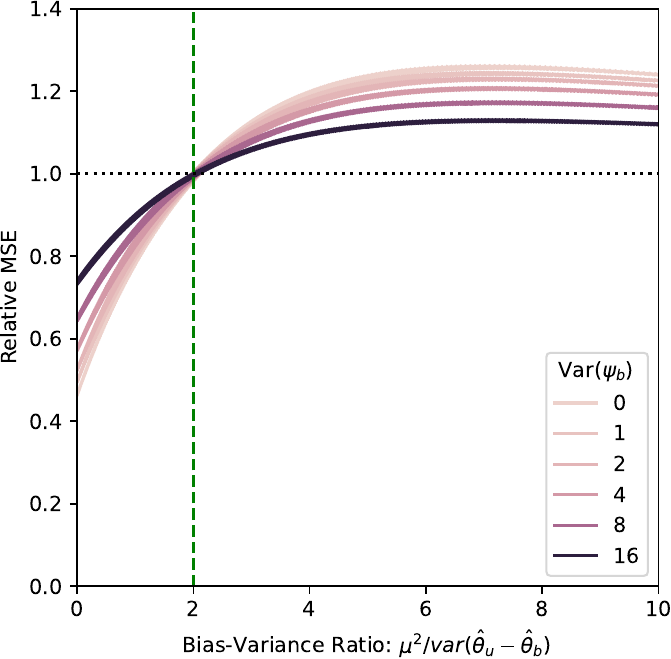}
  \caption{Zero Correlation}%
  \label{subfig:independent}%
  \label{fig:relativeMSE_biasVarianceRatio_indep}
  \end{subfigure}%
  \begin{subfigure}[h]{0.3\linewidth}
  \centering
  \includegraphics[width=\linewidth]{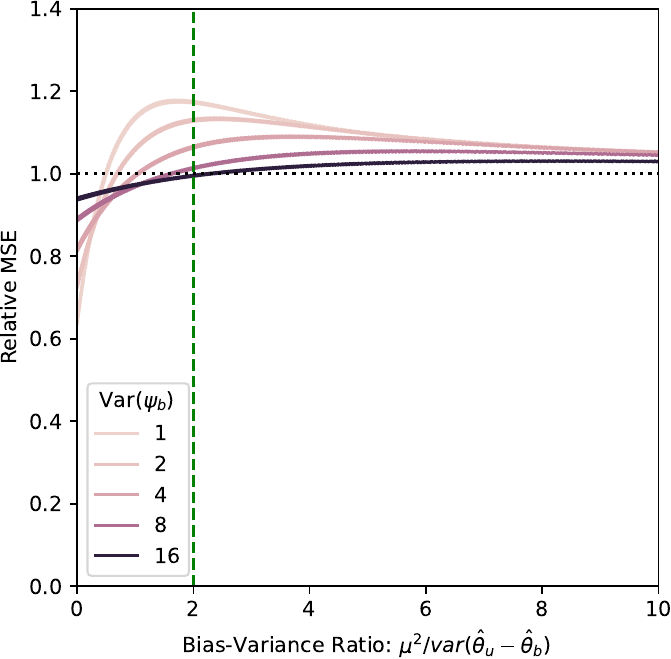}
  \caption{Positive Correlation}%
  \label{subfig:pos_corr}
  \end{subfigure}%
  \begin{subfigure}[h]{0.3\linewidth}
  \centering
  \includegraphics[width=\linewidth]{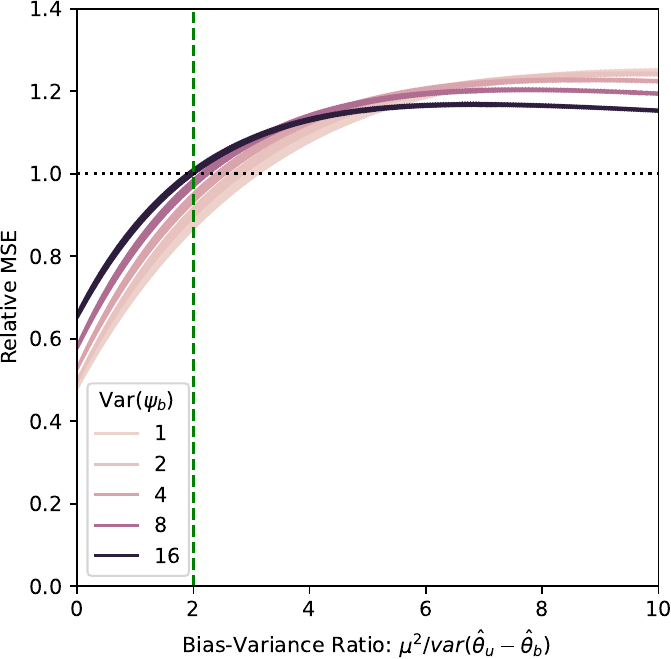}
  \caption{Negative Correlation}%
  \label{subfig:neg_corr}
  \end{subfigure}
  \caption{Relative MSE as a function of the ratio $\mu^2 / \var(\tuh - \tbh)$, across all sample sizes.  The dashed green vertical line denotes $\mu^2 = 2 \var(\tuh - \tbh)$.  (\subref{subfig:independent}) $\text{corr}(\psib, \psiu) = 0$, shown here with $\var(\psiu) = 4$. (\subref{subfig:pos_corr}) Variance fixed at $\var(\psiu) = 16$ and correlation is fixed at $0.5$.  (\subref{subfig:neg_corr}) Same as~\subref{subfig:pos_corr}, but correlation fixed at $-0.5$.  We observe that when the correlation is positive, the maximum tolerable bias is lower, and higher when the correlation is negative, with the gap to predicted threshold determined by the difference in the variance of $\tbh, \tuh$.}%
\label{fig:relativeMSE_biasVarianceRatio_pos_neg_corr}
\end{figure*}

\section{Additional SPRINT details and results}%
\label{sec:additional_sprint_app_section}

\subsection{SPRINT Simulation}%
\label{sec:sprint_simulation_details}

\textbf{Details on Generative Model}: The generative model for potential outcomes in the simulated RCT can be described as follows, consistent with data reported in \citet{SPRINT_Research_Group2015-gb}.
\begin{align*}
  \P(Y_1 = 1 \mid U = 1) &= 0.081 & \P(Y_1 = 1 \mid U = 0) &= 0.040 \\
  \P(Y_0 = 1 \mid U = 1) &= 0.096 & \P(Y_0 = 1 \mid U = 0) &= 0.057 
\end{align*}
Here, the p-value for a heterogeneous treatment effect was not significant ($p=0.32$), but $U = 1$ has a strong marginal association with the primary outcome. The marginal rate of $U$ in the RCT is $1330+1316 / 9361 \approx 28\%$, which we use as our incidence of $U$ across both the simulated RCT and simulated observational study.
\begin{equation*}
  \P(U = 1) = 0.28.
\end{equation*}

\textbf{Details of unbiased/biased estimators}: The estimators $\tuh, \tbh$ are constructed as 
\begin{align*}
  \tuh &= \frac{1}{\nexp} \sum_{i: D_i = E} Y_i \left(\frac{T_i}{\hat{e}_E} - \frac{(1 - T_i)}{1 - \hat{e}_E}\right) & \tbh &= \frac{1}{\nobs} \sum_{i: D_i = O} Y_i \left(\frac{T_i}{\hat{e}_O} - \frac{(1 - T_i)}{1 - \hat{e}_O}\right) 
\end{align*}
where $\hat{e}_d = {(\sum_{i} \1{D_i = d})}^{-1} \sum_{i} T_i \1{D_i = d}$ is an empirical estimate of the treatment probability in dataset $d$. To construct $\thlh$ we estimate $\sigmauhs, \sigmabhs$ by the variance of plug-in estimates of the corresponding influence functions, as described in~\Cref{sec:sprint_simulation_details}. These are used to construct $\thlh$ for each pair of observational and experimental estimators, using 
\begin{equation}
  \lambdah = \frac{\sigmauhs}{{(\tuh - \tbh)}^2 + \sigmauhs + \sigmabhs}.
\end{equation}  
\textbf{Details on Variance Estimation}: The variance of each estimator is estimated as
\begin{align*}
  \sigmauhs &= \frac{1}{\nexp^2} \sum_{i: D_i = E} {\left( (Y_i - \hat{\mu}_E(T_i)) \left(\frac{T_i}{\hat{e}_E} - \frac{(1 - T_i)}{1 - \hat{e}_E}\right) + (\hat{\mu}_E(1) - \hat{\mu}_E(0)) - \tuh \right)}^2 \\
\sigmabhs &= \frac{1}{\nobs^2} \sum_{i: D_i = O} {\left((Y_i - \hat{\mu}_O(T_i)) \cdot \left(\frac{T_i}{\hat{e}_O} - \frac{(1 - T_i)}{1 - \hat{e}_O}\right) + \hat{\mu}_O(1) - \hat{\mu}_O(0) - \tbh \right)}^2
\end{align*}
where $\hat{\mu}_d(t) = {(\sum_{i} \1{D_i = d, T_i = t})}^{-1} \sum_{i} Y_i \1{D_i = d, T_i = t}$ is the empirical mean in treatment arm $t$ in dataset $d$. 

\subsection{Additional SPRINT results}%
\label{sec:additional_sprint_results}

In~\cref{tab:sprint_simulation_gamma_full} we repeat the setup of~\cref{sec:semirealistic_synthetic_experiment} and additionally vary the sample size $\nobs \in \{10000, 20000, 50000, 100000\}$. We observe that the maximum allowable value of $\gamma$ decreases slightly as the sample size increases. 

\begin{table}[t]
  \centering
  \caption{The RMSE of the reference estimator over selected values of $\gamma$. For legibility, the RMSE is multiplied by 1000, on which scale the RMSE of the unbiased estimator is 4.97.  For each sample size we \textbf{bold} the largest value which remains below the RMSE of the unbiased estimator.}%
  \label{tab:sprint_simulation_gamma_full}
  \small
  \begin{tabular}{ccccc}
    \toprule
    $\gamma$ & 10k & 20k & 50k & 100k \\
    \midrule
    0.00 &  4.24 &  3.98 & 3.72 &  3.59 \\
    0.05 &  4.24 &  3.99 & 3.72 &  3.60 \\
    0.10 &  4.24 &  4.00 & 3.74 &  3.62 \\
    0.15 &  4.25 &  4.01 & 3.77 &  3.66 \\
    0.20 &  4.26 &  4.04 & 3.81 &  3.71 \\
    0.25 &  4.28 &  4.07 & 3.86 &  3.77 \\
    0.30 &  4.30 &  4.10 & 3.92 &  3.83 \\
    0.35 &  4.32 &  4.15 & 3.99 &  3.91 \\
    0.40 &  4.35 &  4.19 & 4.06 &  4.00 \\
    0.45 &  4.38 &  4.24 & 4.14 &  4.08 \\
    0.50 &  4.41 &  4.29 & 4.22 &  4.18 \\
    0.55 &  4.44 &  4.35 & 4.30 &  4.27 \\
    0.60 &  4.48 &  4.41 & 4.38 &  4.37 \\
    0.65 &  4.52 &  4.47 & 4.47 &  4.46 \\
    0.70 &  4.55 &  4.53 & 4.55 &  4.55 \\
    0.75 &  4.59 &  4.59 & 4.63 &  4.64 \\
    0.80 &  4.63 &  4.65 & 4.71 &  4.73 \\
    0.85 &  4.67 &  4.70 & 4.78 &  4.81 \\
    0.90 &  4.71 &  4.76 & 4.85 &  4.89 \\
    0.95 &  4.74 &  4.82 & \textbf{4.92} &  \textbf{4.97} \\
    1.00 &  4.78 &  4.87 & 4.99 &  5.04 \\
    1.05 &  4.82 &  4.92 & 5.05 &  5.10 \\
    1.10 &  4.85 &  \textbf{4.97} & 5.10 &  5.16 \\
    1.15 &  4.88 &  5.01 & 5.15 &  5.21 \\
    1.20 &  4.91 &  5.05 & 5.20 &  5.26 \\
    1.25 &  \textbf{4.94} &  5.09 & 5.24 &  5.31 \\
    1.30 &  4.98 &  5.13 & 5.28 &  5.34 \\
    1.35 &  5.00 &  5.16 & 5.31 &  5.38 \\
    1.40 &  5.03 &  5.19 & 5.34 &  5.41 \\
    1.45 &  5.05 &  5.22 & 5.36 &  5.43 \\
    1.50 &  5.07 &  5.24 & 5.39 &  5.45 \\
    1.55 &  5.10 &  5.27 & 5.40 &  5.47 \\
    1.60 &  5.11 &  5.29 & 5.42 &  5.48 \\
    1.65 &  5.13 &  5.30 & 5.43 &  5.49 \\
    1.70 &  5.15 &  5.32 & 5.44 &  5.49 \\
    1.75 &  5.17 &  5.33 & 5.44 &  5.50 \\
    1.80 &  5.18 &  5.34 & 5.45 &  5.50 \\
    1.85 &  5.19 &  5.35 & 5.45 &  5.49 \\
    1.90 &  5.20 &  5.35 & 5.45 &  5.49 \\
    1.95 &  5.21 &  5.36 & 5.45 &  5.49 \\
    2.00 &  5.22 &  5.36 & 5.45 &  5.48 \\
  \bottomrule
  \end{tabular}
\end{table}

\clearpage
\bibliographystyle{icml2022}
\bibliography{references_notes}

\end{document}

%% file: packages.tex
\usepackage{csquotes} 
\usepackage{subcaption}
\usepackage{commath}
\usepackage[inline]{enumitem} 
\usepackage{booktabs}
\usepackage[usenames,dvipsnames]{xcolor}
\usepackage{wrapfig}
\usepackage{comment}

\usepackage{tikz}
\usepackage{tikz-qtree}
\usetikzlibrary{bayesnet}

\usepackage{hyperref}
\hypersetup{
  colorlinks=true,
  linkcolor=RawSienna,
  citecolor=gray
}

\usepackage{amssymb,amsmath,amsthm,amsfonts,mathtools,dsfont,thmtools}

\theoremstyle{definition}
\newtheorem{definition}{Definition}
\newtheorem{assumption}{Assumption}
\newtheorem{example}{Example}
\theoremstyle{remark}

\usepackage[noabbrev,capitalize]{cleveref}

%% file: paperSpecificNotation.tex
\newcommand{\Dobs}{\mathcal{D}_{\text{obs}}}
\newcommand{\Dexp}{\mathcal{D}_{\text{exp}}}
\newcommand{\nobs}{n_{\text{obs}}}
\newcommand{\nexp}{n_{\text{exp}}}

\newcommand{\corr}{\text{Corr}}

\newcommand{\lambdas}{\lambda^*}

\newcommand{\tauh}{\hat{\tau}}

\newcommand{\sigmah}{\hat{\sigma}}

\newcommand{\sigmau}{\sigma_u}
\newcommand{\sigmab}{\sigma_b}
\newcommand{\sigmaus}{{\sigma}^2_u}
\newcommand{\sigmabs}{{\sigma}^2_b}
\newcommand{\sigmabh}{\hat{\sigma}_b}

\newcommand{\sigmabhs}{\hat{\sigma}^2_b}
\newcommand{\sigmauhs}{\hat{\sigma}^2_u}
\newcommand{\sigmacov}{\sigma_{bu}}
\newcommand{\sigmacovh}{\hat{\sigma}_{bu}}

\newcommand{\thetahlh}{\hat{\theta}_{\hat{\lambda}}}

\newcommand{\tz}{\theta_0}
\newcommand{\tbh}{\thetah_b}
\newcommand{\tuh}{\thetah_u}
\newcommand{\tb}{\theta_b}

\newcommand{\thh}{\hat{\theta}}
\newcommand{\thl}{\hat{\theta}_{\lambda}}
\newcommand{\thlh}{\hat{\theta}_{\hat{\lambda}}}

\newcommand{\infu}{\phi_u}
\newcommand{\infb}{\phi_b}
\newcommand{\infuh}{\hat{\phi}_u}
\newcommand{\infbh}{\hat{\phi}_b}
\newcommand{\psiu}{\psi_u}
\newcommand{\psib}{\psi_b}

\newcommand{\tbhk}{\tbh^{(k)}}
\newcommand{\thetahlhk}{\thetahlh^{(k)}}